\documentclass[11pt]{article}
\newcommand{\myifthen}[2]{#1}

\usepackage{graphicx}
\usepackage{algorithm,algorithmic}
\usepackage{amsmath, amssymb}
\usepackage{prettyref}
\usepackage{pstricks,pst-node,pst-text}
\usepackage{boxedminipage}
\usepackage{xcolor}
\definecolor{colcite}{rgb}{0,0.75,0}
\definecolor{collink}{rgb}{0,0,1}
\definecolor{colurl}{rgb}{1,0,0}
\myifthen{\usepackage[colorlinks=true,citecolor=colcite,linkcolor=collink,urlcolor=colurl]{hyperref}}{}
\usepackage{picins}

\newlength{\lpbox}
\newlength{\algobox}

\myifthen{
  \usepackage{fullpage}
  \usepackage{amsthm}
  \usepackage[letterpaper,margin=1in]{geometry}

  \newtheorem{theorem}{Theorem}
  \newtheorem{nottheorem}{}[section]
  \newrefformat{theorem}{Theorem \ref{#1}}
  \newrefformat{eq}{Equation \eqref{#1}}
  \def\xthm[#1][#2][#3]{\newtheorem{#2}[nottheorem]{#3} \newrefformat{#2}{#3 \ref{#11}}}
  \xthm[#][definition][Definition]
  \xthm[#][claim][Claim]
  \xthm[#][prop][Proposition]
  \xthm[#][ppty][Property]
  \xthm[#][remark][Remark]
  \xthm[#][lemma][Lemma]
  \xthm[#][example][Example]
  \xthm[#][fact][Fact]
  \xthm[#][corollary][Corollary]
  \xthm[#][app][Appendix]

  \setlength{\lpbox}{8cm}
  \setlength{\algobox}{6.5in}
}{
  \setlength{\lpbox}{6cm}
  \setlength{\algobox}{4.8in}

  \def\xthm[#1][#2][#3]{\newtheorem{#2}[nottheorem]{#3}\newrefformat{#2}{#3 \ref{#11}}}
  \xthm[#][ppty][Property]

  \newrefformat{theorem}{Theorem \ref{#1}}
  \newrefformat{eq}{Equation \eqref{#1}}
  \newrefformat{fig}{Figure \ref{#1}}
  \newrefformat{remark}{Remark \ref{#1}}
}

\def\U{\mathcal U}
\def\7{\frac{73}{60}}
\def\Shrink{{\tt Shrink}}
\def\L{\mathcal L}

\newcommand{\steiner}[2]{\fnode[framesize=0.2](#1){#2}}
\newcommand{\terminal}[2]{\cnode[fillstyle=solid,fillcolor=black](#1){.12}{#2}}
\newcommand{\bs}{\backslash}
\newcommand{\comment}[1]{}
\newcommand{\ignore}[1]{}
\newcommand{\NP}{\ensuremath{\mathsf{NP}}}
\newcommand{\PP}{\ensuremath{\mathsf{P}}}
\newcommand{\APX}{\ensuremath{\mathsf{APX}}}

\newcommand{\bdp}{\eqref{eq:LP-P2}}
\newcommand{\mtst}{\ensuremath{{\mathtt{mtst}}}}

\def\K{\ensuremath{\mathcal{K}}}

\def\C{\ensuremath{\mathcal{C}}}

\def\Z{\mathbf{Z}}
\def\R{\mathbf{R}}

\def\U{\mathcal U}

\def\rc{\mathtt{rc}}
\def\cp{\mathtt{cp}}

\def\dout{\delta^{\mathrm{out}}}
\newcommand{\spn}{\mathop{{\tt span}}}

\newcommand{\supp}{\mathop{{\tt supp}}}
\newcommand{\tight}{\mathop{{\tt tight}}}
\newcommand{\OPT}{\ensuremath{\mathop{\mathrm{OPT}}}}

\def\F{\mathcal F}
\def\Z{\mathcal Z}
\def\X{\mathcal X}
\def\valid{\mathrm{valid}}
\def\A{\bf A}
\def\mss{\textsc{RatioGreedy}}

\newcounter{thm_locopt}
\newcounter{thm_saved}

\title{Hypergraphic LP Relaxations for Steiner Trees}

\myifthen{
  \date{University of Waterloo
    \footnote{Supported by NSERC grant no.\ 288340 and by an Early Research
      Award. Email: (deepc, jochen, dagpritc @uwaterloo.ca)}}

  \author{Deeparnab Chakrabarty \and Jochen K\"onemann \and
    David Pritchard}
}{
 \author{
    Deeparnab Chakrabarty
    \and
    Jochen K\"onemann
    \and
    David Pritchard
  }

 \institute{
    University of Waterloo, Waterloo, Ontario N2L 3G1, Canada\thanks{Supported by NSERC grant no.\ 288340
      and by an Early Research Award}
  }
}

\begin{document}

\maketitle

\begin{abstract}
  We investigate hypergraphic LP relaxations for the Steiner tree
  problem, primarily the partition LP relaxation introduced by
  K\"{o}ne\-mann et al.~[Math.\ Programming, 2009]. Specifically, we are
  interested in proving upper bounds on the integrality gap of this
  LP, and studying its relation to other linear relaxations.  Our
  results are the following.

  \smallskip

  \noindent{\bf Structural results:} We extend the technique of
  uncrossing, usually applied to families of sets, to families of
  partitions.  As a consequence we show that any basic feasible
  solution to the partition LP formulation has sparse support.
  Although the number of variables could be exponential, the number of
  positive variables is at most the number of terminals.

  \smallskip

  \noindent{\bf Relations with other relaxations:} We show the
  equivalence of the partition LP relaxation with other known
  hypergraphic relaxations.  We also show that these hypergraphic
  relaxations are equivalent to the well studied bidirected cut
  relaxation, if the instance is quasibipartite.

  \smallskip

  \noindent{\bf Integrality gap upper bounds:} We show an upper bound
  of $\sqrt{3} \doteq 1.729$ on the integrality gap of these
  hypergraph relaxations in general graphs. In the special case of
  uniformly quasibipartite instances, we show an improved upper bound
  of $73/60 \doteq 1.216$. By our equivalence theorem, the latter
  result implies an improved upper bound for the bidirected cut
  relaxation as well.
\end{abstract}

\section{Introduction}\label{sec:intro}

In the {\em Steiner tree} problem, we are given an undirected graph
$G=(V,E)$, non-negative costs $c_e$ for all edges $e \in E$, and a set
of {\em terminal} vertices $R \subseteq V$. The goal is to find a
minimum-cost tree $T$ spanning $R$, and possibly some {\em Steiner
  vertices} from $V\setminus R$. We can assume that the graph is
complete and that the costs induce a metric.  The problem takes a
central place in the theory of combinatorial optimization and has
numerous practical applications.  Since
the Steiner tree problem is \NP-hard\footnote{Chleb{\'i}k and
  Chleb{\'i}kov{\'a} show that no $(96/95-\epsilon)$-approximation
  algorithm can exist for any positive $\epsilon$ unless \PP=\NP
  ~\cite{CC02}.} we are interested in approximation algorithms for it.
The best published approximation algorithm for the Steiner tree
problem is due to Robins and Zelikovsky \cite{RZ05}, which for any
fixed $\epsilon > 0$, achieves a performance ratio of $1+\frac{\ln
  3}{2}+\epsilon \doteq 1.55$ in polynomial time; an improvement is
currently in press~\cite{BGRS10}, see also \prettyref{remark:byrka}.

In this paper, we study linear programming (LP) relaxations for the
Steiner tree problem, and their properties. Numerous such formulations
are known \myifthen{(e.g., see
\cite{An80,CR94a,CR94b,DB+01,Ed67,GM93,Pol03,PVd01,War98,Wo84}),}{(e.g., see \cite{Ed67,GM93,Pol03,PVd01,War98,Wo84}),} and
their study has led to impressive running time improvements for
integer programming based methods. Despite the significant body of
work in this area, none of the known relaxations is known to exhibit
an {\em integrality gap} provably smaller\myifthen{\footnote{Achieving an integrality gap
  of $2$ is relatively easy for most relaxations by showing that the
  minimum spanning tree restricted on the terminals is within a factor
  $2$ of the LP.}}{}
than $2$. The integrality gap of a relaxation is the maximum
ratio of the cost of integral and fractional optima, over all
instances. It is commonly regarded as a measure of strength of
a formulation. One of the contributions of this paper are
improved bounds on the integrality gap for a number of Steiner
tree LP relaxations.

A Steiner tree relaxation of particular interest is the {\em
  bidirected cut relaxation} \cite{Ed67,Wo84} (precise definitions
will follow in \prettyref{sec:lpss}). This relaxation has a flow formulation
using $O(|E||R|)$ variables and constraints, which is much more compact
than the other relaxations we study. %DP: others can too! and hence can be solved efficiently.
Also, it is also widely believed to have an
integrality gap significantly smaller than $2$ (e.g., see
\cite{CDV08,RV99,Va00}). The largest lower bound on the integrality
gap known is $8/7$ (by Martin Skutella, reported in \cite{KPT09}), and
Chakrabarty et al. \cite{CDV08} prove an upper bound of $4/3$ in so
called {\em quasi-bipartite} instances (where Steiner vertices form an
independent set).

Another class of formulations are the so called {\em hypergraphic} LP
relaxations for the Steiner tree problem. These relaxations are
inspired by the observation that the minimum Steiner tree problem can
be encoded as a minimum cost hyper-spanning tree (see
\prettyref{sec:hyp}) of a certain hypergraph on the terminals.  They
are known to be stronger than the bidirected cut
relaxation~\cite{PVd03}, and it is therefore natural to try to use
them to get better approximation algorithms, by drawing on the large
corpus of known LP techniques.
In this paper, we focus on one hypergraphic LP in particular: the
{\em partition} LP of K\"{o}nemann et al.~\cite{KPT09}.

\subsection{Our Results and Techniques}

There are three classes of results in this paper: structural results,
equivalence results, and integrality gap upper bounds.

\smallskip
\noindent {\bf Structural results}, \prettyref{sec:up}: We extend the
powerful technique of {\em uncrossing}, traditionally applied to
families of sets, to families of {\em partitions}. Set uncrossing has
been very successful in obtaining exact and approximate algorithms for
a variety of problems (for instance, \cite{EG77,J98,SL07}). Using
partition uncrossing, we show that any basic feasible solution to the
partition LP has at most $(|R|-1)$ positive variables (even though it
can have an exponentially large number of variables and constraints).

\smallskip
\noindent {\bf Equivalence results}, \prettyref{sec:equiv}: In
addition to the partition LP, two other hypergraphic LPs have been
studied before: one based on \emph{subtour elimination} due to Warme
\cite{War98}, and a \emph{directed hypergraph relaxation} of Polzin
and Vahdati Daneshmand \cite{PVd03}; these two are known to be
equivalent \cite{PVd03}. We prove that in fact \emph{all three
  hypergraphic relaxations are equivalent} (that is, they have the same objective value
  for any Steiner tree instance).
\myifthen{We
give two proofs (for completeness and to demonstrate our new
techniques), one showing the equivalence of  the partition LP and the subtour LP via
partition uncrossing, and one showing the equivalence of the partition LP to the
directed LP via hypergraph orientation results of Frank et
al.~\cite{FKK03b}.}{}

We also show that, on {\em quasibipartite instances}, the
hypergraphic and the bidirected cut LP relaxations are
equivalent.
\myifthen {
  We find this surprising for the following reasons.
  Firstly, some instances are known where the hypergraph relaxations
  is {\em strictly} stronger than the bidirected cut
  relaxation~\cite{PVd03}. Secondly, the bidirected cut relaxations
  seems to resist uncrossing techniques; e.g.\ even in quasi-bipartite
  graphs extreme points for bidirected cut can have as many as
  $\Omega(|V|^2)$ positive
  variables~\cite[Sec.~4.9]{P09thesis}. Thirdly, the known approaches
  to exploiting the bidirected cut relaxation (mostly primal-dual and
  local search algorithms \cite{RV99,CDV08}) are very different from
  the combinatorial hypergraphic algorithms for the Steiner tree
  problem (almost all of them employ greedy strategies). In short,
  there is no qualitative similarity to suggest why the two
  relaxations should be equivalent!
}{
  This result is surprising since we are aware of no qualitative similarity to suggest why the two
  relaxations should be equivalent.
}
We believe a better understanding
  of the bidirected cut relaxation is important because it is central in theory \emph{and}
  practical for implementation.

\smallskip
\noindent {\bf Improved integrality gap upper bounds},
\prettyref{sec:gapbounds}: For \emph{uniformly quasibipartite
  instances} (quasibipartite instances where for each Steiner vertex,
all incident edges have the same cost), we show that the integrality
gap of the hypergraphic LP relaxations is upper bounded by $73/60
\doteq 1.216$.  Our proof uses the approximation algorithm of
Gr\"{o}pl et al.~\cite{GH+02} which achieves the same ratio with
respect to the (integral) optimum. We show, via a simple dual fitting
argument, that this ratio is also valid with respect to the LP
value. To the best of our knowledge this is the only nontrivial class
of instances where the best currently known approximation ratio and
integrality gap upper bound are the same.

For general graphs, we give simple upper bounds of $2\sqrt{2}-1 \doteq
1.83$ and $\sqrt{3} \doteq 1.729$ on the integrality gap of the
hypergraph relaxation. Call a graph {\em gainless} if the minimum {\em
  spanning} tree of the terminals is the optimal Steiner tree. To
obtain these integrality gap upper bounds, we use the following key
property of the hypergraphic relaxation which was implicit in
\cite{KPT09}: on gainless instances (instances where the optimum terminal
spanning tree is the optimal Steiner tree), the LP value equals the minimum
spanning tree and the integrality gap is 1. Such a theorem was known
for quasibipartite instances and the bidirected cut relaxation
(implicitly in \cite{RV99}, explicitly in \cite{CDV08}); we
extend techniques of \cite{CDV08} to obtain improved integrality gaps
on all instances.

\begin{remark}\label{remark:byrka}
  The recent independent work of Byrka et al.~\cite{BGRS10}, which gives
  an improved approximation for Steiner trees in general graphs, also shows an
  integrality gap bound of $1.55$ on the hypergraphic directed cut
  LP.
  This is stronger than our integrality gap bounds and was obtained prior
to the completion of our paper;
yet we include our bounds because they are obtained
  using fairly different methods which might be of independent
  interest in certain settings.

  \comment{ We have added these bounds to this paper despite the
    existence of a stronger result for two main reasons. First, the
    methods used in our proofs are independent of the work in
    \cite{BGRS10} and rely on properties of hypergraphic relaxations
    in gainless graphs derived in \cite{KPT09}, and the scaling
    technique of \cite{CDV08}. Secondly, the algorithms to prove these
    upper bounds are simple, and in addition have an ``online''
    flavour, which we believe is interesting in its own right.}

  \comment{ (see the e-print \cite{CKP09}. The result of the upper
    bound of $\sqrt{3}$ in general graphs was obtained afterwards. We
    have added this result to our paper despite the existence of a
    stronger result because of two reasons. We believe the theorem of
    \cite{KPT09} about hypergraphic relaxations in gainless graphs,
    and the scaling technique of \cite{CDV08} are simple and, as we
    show, quite useful in proving upper bounds. Secondly, the
    algorithm to prove the upper bound is a simple algorithm (along
    the lines of that in \cite{CDV08}) which we believe deserves merit
    in itself.  We stress here that although we were aware of the
    results of \cite{BGRS10}, our methods are completely independent.}

  The proof in \cite{BGRS10} can be easily modified to show an
  integrality gap upper bound of $1.28$ in quasibipartite
  instances. Then using our equivalence result, we get an integrality
  gap upper bound of $1.28$ for the bidirected cut relaxation on
  quasibipartite instances, improving the previous best of $4/3$.
\end{remark}

\subsection{Bidirected Cut and Hypergraphic Relaxations}\label{sec:lpss}
\subsubsection{The Bidirected Cut Relaxation}\label{sec:bcr}

The first bidirected LP was given by Edmonds \cite{Ed67} as an exact
formulation for the spanning tree problem. Wong \cite{Wo84} later
extended this to obtain the bidirected cut relaxation for the Steiner
tree problem, and gave a dual ascent heuristic based on the
relaxation.  For this relaxation, introduce two arcs $(u,v)$ and
$(v,u)$ for each edge $uv \in E$, and let both of their costs be
$c_{uv}$. Fix an arbitrary terminal $r \in R$ as the root.  Call a
subset $U \subseteq V$ {\em valid} if it contains a terminal but not
the root, and let $\valid(V)$ be the family of all valid sets.
Clearly, the in-tree rooted at $r$ (the directed tree with all
vertices but the root having out-degree exactly $1$) of a Steiner tree
$T$ must have at least one arc with tail in $U$ and head outside $U$,
for all valid $U$.  This leads to the bidirected cut relaxation
\eqref{eq:LP-B} (shown in \prettyref{fig:B} with dual) which has a
variable for each arc $a \in A$, and a constraint for every valid set
$U$.  Here and later, $\dout(U)$ denotes the set of arcs in $A$ whose
tail is in $U$ and whose head lies in $V\setminus U$. When there are
no Steiner vertices, Edmonds' work~\cite{Ed67} implies this relaxation
is exact.

\begin{figure}[h]
  \begin{minipage}{\lpbox} \begin{align}
      \min \sum_{a \in A} c_ax_a: \quad& x \in \R^A_{\ge 0}
\tag{\ensuremath{\mathcal{B}}}\label{eq:LP-B} \\
      \sum_{a\in \dout(U)} x_a \ge 1, \quad& \forall
       U \in {\valid(V)}  \myifthen{\label{eq:LP-B2}}{\notag}
    \end{align} \end{minipage}
  \hfill \vline \hfill
  \begin{minipage}{\lpbox} \begin{align}
      \max \sum_{U} z_U: \quad& z \in
\R^{\valid(V)}_{\ge 0}
\tag{\ensuremath{\mathcal{B}_D}}\label{eq:LP-BD} \\
      \sum_{U:a\in \dout(U)} z_U \le c_a, \quad&\forall a\in A   \myifthen{\label{eq:LP-BD1}}{\notag}
    \end{align}\end{minipage}
  \caption{\small The bidirected cut relaxation \eqref{eq:LP-B} and its dual \eqref{eq:LP-BD}.}\label{fig:B}
\end{figure}

Goemans \& Myung~\cite{GM93} made significant progress in
understanding the LP, by showing that the bidirected cut LP has the
same value independent of which terminal is chosen as the root, and by
showing that a whole ``catalogue" of very different-looking LPs also
has the same value; later Goemans~\cite{Go94} showed that if the graph
is series-parallel, the relaxation is exact. Rajagopalan and Vazirani
\cite{RV99} were the first to show a non-trivial integrality gap upper
bound of $3/2$ on quasibipartite graphs; this was subsequently
improved to $4/3$ by Chakrabarty et al.~\cite{CDV08}, who gave another
alternate formulation for \eqref{eq:LP-B}.

\subsubsection{Hypergraphic Relaxations}\label{sec:hyp}
Given a Steiner tree $T$, a \emph{full component} of $T$ is a maximal
subtree of $T$ all of whose leaves are terminals and all of whose
internal nodes are Steiner nodes. The edge set of any Steiner tree can
be partitioned in a {\em unique} way into full components by splitting
at internal terminals; see \prettyref{fig:decomp} for an example.

\begin{figure}[htb]
  \begin{center} \leavevmode
    \myifthen{\begin{pspicture}(0,0)(4,2.4)\psset{unit=0.8}}{\begin{pspicture}(0,0)(3.8,2)\psset{unit=0.65}}
      \terminal{0,2}{t1}
      \terminal{0.6,0.8}{t2}
      \steiner{0.6,1.3}{s1}
      \ncline{s1}{t1}
      \ncline{s1}{t2}
      \steiner{1.05,2.45}{s2}
      \steiner{1.45,0.65}{s3}
      \steiner{1.21,1.55}{s4}
      \ncline{s1}{s4}
      \terminal{1.8,2}{t3}
      \ncline{s4}{t3}
      \terminal{2.4,1}{t4}
      \ncline{s4}{t4}
      \terminal{2.5,0.1}{t5}
      \steiner{2.46,2.2}{s5}
      \ncline{t4}{t5}
      \steiner{3.6,1.45}{s6}
      \terminal{3.6,2.88}{t6}
      \ncline{t4}{s6}
      \ncline{s6}{t6}
       \terminal{4.7,0.9}{t8}
      \terminal{4.6,0.2}{t9}
      \ncline{s6}{t8}
      \ncline{t8}{t9}
      \steiner{4.2,0.56}{s7}
    \end{pspicture}
    \myifthen{\hfill}{}
    \myifthen{\begin{pspicture}(0,0)(4,2.4)\psset{unit=0.8}}{\begin{pspicture}(0,0)(3.8,2)\psset{unit=0.65}}
        \terminal{0,2}{t1}
        \terminal{0.6,0.8}{t2}
        \steiner{0.6,1.3}{s1}
        \ncline{s1}{t1}
        \ncline{s1}{t2}
        \steiner{1.21,1.55}{s4}
        \ncline{s1}{s4}
        \terminal{1.8,2}{t3}
        \ncline{s4}{t3}
        \terminal{2.4,1}{t4l}
        \terminal{2.8,1}{t4r}
        \terminal{2.6,0.75}{t4b}
        \ncline{s4}{t4l}
        \terminal{2.7,-0.15}{t5}
        \ncline{t4b}{t5}
        \steiner{4.0,1.45}{s6}
        \terminal{4.0,2.88}{t6}
        \ncline{t4r}{s6}
        \ncline{s6}{t6}
        \terminal{5.1,0.9}{t8l}
        \terminal{5.1,0.6}{t8r}
        \terminal{5.0,-0.1}{t9}
        \ncline{s6}{t8l}
        \ncline{t8r}{t9}
      \end{pspicture}
      \myifthen{\hfill}{}
      \myifthen{\begin{pspicture}(0,0)(4,2.4)\psset{unit=0.8}}{\begin{pspicture}(0,0)(3.8,2)\psset{unit=0.65}}
          \pspolygon[linearc=.2](-0.3,2.2)(0.5,0.6)(2.7,0.84)(1.9,2.2)
          \pspolygon[linearc=.2](2.1,0.8)(5,0.7)(3.6,3.3)
          \psellipse(2.45,0.55)(0.3,0.75)
          \psellipse(4.65,0.55)(0.3,0.75)
          \terminal{0,2}{t1}
          \terminal{0.6,0.8}{t2}
          \terminal{1.8,2}{t3}
          \terminal{2.4,1}{t4}
          \terminal{2.5,0.1}{t5}
          \terminal{3.6,2.88}{t6}
          \terminal{4.7,0.9}{t8}
          \terminal{4.6,0.2}{t9}
        \end{pspicture}
      \end{center}
      \caption{\small Black nodes are terminals and white nodes are
        Steiner nodes. Left: a Steiner tree for this instance.
        Middle: the Steiner tree's edges are partitioned into full
        components; there are four full components. Right: the
        hyperedges corresponding to these full
        components.} \label{fig:decomp} \end{figure}
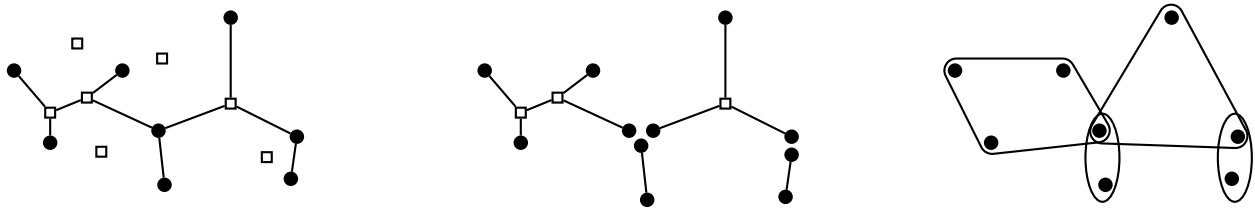

Let $\K$ be the set of all nonempty subsets of terminals
(\emph{hyperedges}). We associate with each $K \in \K$ a fixed full
component spanning the terminals in $K$, and let $C_K$ be its
cost\footnote{We choose the minimum cost full component if there are
  many. If there is no full component spanning $K$, we let $C_K$ be
  infinity. Such a minimum cost component can be found in polynomial
  time, if $|K|$ is a constant.}.  The problem of finding a
minimum-cost Steiner tree spanning $R$ now reduces to that of finding
a minimum-cost hyper-spanning tree in the hypergraph $(R,\K)$.

Spanning trees in (normal) graphs are well understood and there are
many different exact LP relaxations for this problem.  These exact LP
relaxations for spanning trees in graphs inspire the {\em hypergraphic
  relaxations} for the Steiner tree problem. Such relaxations have a
variable $x_K$ for every\footnote{Observe that there could be
  exponentially many hyperedges.  This computational issue is
  circumvented by considering hyperedges of size at most $r$, for some
  constant $r$. By a result of Borchers and Du~\cite{BD97}, this leads
  to only a $(1+\Theta(1/\log r))$ factor increase in the optimal
  Steiner tree cost.} $K\in \K$, and the different relaxations are
based on the constraints used to capture a hyper-spanning tree, just
as constraints on edges are used to capture a spanning tree in a
graph.

The oldest hypergraphic LP relaxation is the subtour LP introduced by
Warme \cite{War98} which is inspired by Edmonds' subtour elimination
LP relaxation \cite{Ed71} for the spanning tree polytope. This LP
relaxation uses the fact that there are no hypercycles in a
hyper-spanning tree, and that it is spanning. More formally, let
$\rho(X) := \max(0,|X|-1)$ be the {\em rank} of a set $X$ of
vertices. Then a sub-hypergraph $(R,\K')$ is a hyper-spanning tree iff
$\sum_{K \in \K'} \rho(K)=\rho(R)$ and $\sum_{K \in \K'} \rho (K \cap
S) \leq \rho(S)$ for every subset $S$ of $R$. The corresponding LP
relaxation, denoted below as \eqref{eq:LP-S}, is called the {\em subtour elimination} LP relaxation.

\begin{align}
  \min \Big\{\sum_{K \in \K} C_Kx_K: ~ & x \in \R^\K_{\ge 0}, ~
  \sum_{K \in \K} x_K\rho(K) = \rho(R), \tag{\ensuremath{\mathcal{S}}}\label{eq:LP-S} \\
  & \sum_{K \in \K} x_K\rho(K \cap S) \leq \rho(S), ~\forall S
  \subset R \Big\} \notag
\end{align}

Warme showed that if the maximum hyperedge size $r$ is bounded by a
constant, the LP can be solved in polynomial time.

\def\Delin{\Delta^{\mbox{\scriptsize {\ensuremath{\mathrm{in}}}}}}
\def\Delout{\Delta^{\mbox{\scriptsize {\ensuremath{\mathrm{out}}}}}}

The next hypergraphic LP introduced for Steiner tree was a directed
hypergraph formulation \eqref{eq:LP-PUDir}, introduced by Polzin and
Vahdati Daneshmand \cite{PVd03}, and inspired by the bidirected cut
relaxation. Given a full component $K$ and a terminal $i\in K$, let
$K^i$ denote the arborescence obtained by directing all the edges of
$K$ towards $i$. Think of this as directing the hyperedge $K$ towards
$i$ to get the directed hyperedge $K^i$. Vertex $i$ is called the
\emph{head} of $K^i$ while the terminals in $K\setminus i$ are the
\emph{tails} of $K$.  The cost of each directed hyperedge $K^i$ is
the cost of the corresponding undirected hyperedge $K$.  In the
directed hypergraph formulation, there is a variable $x_{K^i}$ for
every directed hyperedge $K^i$.  As in the bidirected cut relaxation,
there is a vertex $r\in R$ which is a root, and as described above, a
subset $U\subseteq R$ of terminals is valid if it does not contain the
root but contains at least one vertex in $R$. We let $\Delout(U)$ be
the set of directed full components coming out of $U$, that is all
$K^i$ such that $U\cap K\neq \varnothing$ but $i\notin U$.  Let
$\overrightarrow{\K}$ be the set of all directed hyperedges. We show
the directed hypergraph relaxation and its dual in Figure \ref{fig:PUDir}.
%\prettyref{fig:PUDir}.

\begin{figure}[h]
  \begin{minipage}{\lpbox}
    \begin{align} \min \Big\{\sum_{K \in \K,i\in
        K} C_{K}x_{K^i}: & \,\, x \in \R^{\overrightarrow{\K}}_{\ge 0}
      \tag{\ensuremath{\mathcal{D}}}\label{eq:LP-PUDir} \\
      \sum_{K^i \in \Delout(U)} x_{K^i} \geq 1, \quad& \forall \mbox{
        valid }~ U\subseteq R
       \Big\} \myifthen{\label{eq:LP-PUDir2}}{\notag}
    \end{align} \end{minipage}
  \hfill \vline \hfill
  \begin{minipage}{\lpbox} \begin{align}
      \max \Big\{\sum_{U} z_U:~~~~ \quad \myifthen{}{\hspace{1cm}}&\myifthen{}{\hspace{-1cm}}z \in
      \R^{\textrm{valid}(R)}_{\ge 0}\!\!\!\! \tag{\ensuremath{\mathcal{D}_D}}\label{eq:LP-A} \\
      \sum_{U:K \cap U \neq \varnothing, i \notin U} z_U \leq C_K,
      \quad&\forall K\in \K, \myifthen{\forall}{} i\in K \Big\}\myifthen{}{\!\!} \myifthen{\label{eq:LP-A2}}{\notag}
    \end{align}\end{minipage}
  \caption{\small The directed hypergraph relaxation \eqref{eq:LP-PUDir} and its
dual \eqref{eq:LP-A}.}\label{fig:PUDir} \end{figure}

Polzin \& Vahdati Daneshmand~\cite{PVd03} showed that
$\OPT\eqref{eq:LP-PUDir}=\OPT\eqref{eq:LP-S}$. Moreover they observed
that this directed hypergraphic relaxation strengthens the bidirected
cut relaxation.
\begin{lemma}[\cite{PVd03}]\label{lem:simple}
  For any instance, $\OPT\eqref{eq:LP-PUDir} \ge \OPT\eqref{eq:LP-B}$.
  \myifthen{}{There are instances for which this inequality is strict.}
\end{lemma}
\myifthen{
  \begin{proof}[Proof sketch.]
    It suffices to show that any solution $x$ of \eqref{eq:LP-PUDir} can
    be converted to a feasible solution $x'$ of \eqref{eq:LP-B} of the
    same cost. For each arc $a$, let $x'_a$ be the sum of $x_{K^i}$ over
    all directed full components $K^i$ that (when viewed as an
    arborescence) contain $a$. Now for any valid subset $U$ of $V$, it
    is not hard to see that every directed full component leaving $R
    \cap U$ has at least one arc leaving $U$, hence $\sum_{a\in
      \dout(U)} {x'}_a \ge \sum_{K^i \in \Delout(R \cap U)} x_{K^i} \ge
    1$ and $x'$ is feasible as needed.
  \end{proof}
  \noindent See \cite{PVd03} for an example where the strict inequality
  $\OPT\eqref{eq:LP-PUDir} > \OPT\eqref{eq:LP-B}$ holds.
}{}

K\"onemann et al.~\cite{KPT09}, inspired by the work of Chopra
\cite{Cho89}, described a partition-based relaxation which captures
that given any partition of the terminals, any hyper-spanning tree
must have sufficiently many ``cross hyperedges". More formally, a
partition, $\pi$, is a collection of pairwise disjoint nonempty
terminal sets $(\pi_1, \ldots, \pi_q)$ whose union equals $R$. The
number of {\em parts} $q$ of $\pi$ is referred to as the partition's
{\em rank} and denoted as $r(\pi)$.  Let $\Pi_R$ be the set of all
partitions of $R$. Given a partition $\pi = \{\pi_1, \ldots, \pi_q\}$,
define the {\em rank contribution} $\rc_K^\pi$ of hyperedge $K \in \K$
for $\pi$ as the rank reduction of $\pi$ obtained by merging the parts
of $\pi$ that are touched by $K$; i.e., $\rc_K^{\pi} := |\{i \,:\,
K\cap \pi_i \neq \varnothing \}| - 1.$ Then a hyper-spanning tree
$(R,\K')$ must satisfy $\sum_{K\in \K'} \rc^\pi_K \ge r(\pi) - 1$.
The partition based LP of \cite{KPT09} and its dual are given in
\prettyref{fig:PU}.

\begin{figure}[h]
  \begin{minipage}{\lpbox} \begin{align}
      \min \Big\{\sum_{K \in \K} C_Kx_K: \quad& x \in \R^\K_{\ge 0}
\tag{\ensuremath{\mathcal{P}}}\label{eq:LP-PU} \myifthen{\!\!\!}{} \\
      \sum_{K \in \K} x_K\rc_K^\pi \geq r(\pi)-1, \quad& \forall \pi
\in \Pi_R \Big\} \myifthen{\label{eq:LP-PU2}}{\notag}
    \end{align} \end{minipage}
  \hfill \vline \hfill
  \begin{minipage}{\lpbox} \begin{align}
      \max \Big\{\sum_{\pi} (r(\pi)-1)\myifthen{\cdot}{}y_\pi: \quad& y \in
\R^{\Pi_R}_{\ge 0} \tag{\ensuremath{\mathcal{P}_D}}\label{eq:LP-PUD} \\
      \sum_{\pi \in \Pi_R} y_\pi\rc_K^\pi\leq C_K,
\quad&\forall K\in \K \Big\} \myifthen{\label{eq:LP-PUD3}}{\notag}
\end{align}\end{minipage}
\caption{\small The unbounded partition relaxation \eqref{eq:LP-PU} and its
  dual \eqref{eq:LP-PUD}.}\label{fig:PU}
\end{figure}

The feasible region of \eqref{eq:LP-PU} is \emph{unbounded}, since if
$x$ is a feasible solution for \eqref{eq:LP-PU} then so is any $x' \ge
x$.  We obtain a \emph{bounded} partition LP relaxation, denoted by
\eqref{eq:LP-P2} and shown below, by adding a valid equality
constraint to the LP.

\begin{align}\tag{$\mathcal P'$}\label{eq:LP-P2}
  \min \Big\{\sum_{K \in \K} C_Kx_K: x \in \eqref{eq:LP-PU}, \sum_{K \in \K} x_K (|K|-1) = |R|-1 \Big\}
\end{align}

\myifthen{
\subsubsection{Discussion of Computational Issues}\label{sec:discussion}
The bidirected cut relaxation is very attractive from a perspective of
computational implementation.  Although the formulation given in
\prettyref{sec:bcr} has an exponential number of constraints, an equivalent
compact flow formulation with $O(|E||R|)$ variables and constraints is
well-known.

What is known regarding solving the hypergraphic LPs? They are good
enough to get theoretical results but less attractive in practice, as
we now explain. Using a separation oracle, Warme showed~\cite{War98}
that for any chosen family $\K$ of full components, the subtour LP can
be optimized in time $\textrm{poly}(|V|,|\K|)$. For the common
\emph{$r$-restricted setting} of $\K$ to be all possible full
components of size at most $r$ for constant $r$, we have $\K \le
\tbinom{|R|}{r}$. This is polynomial for any fixed $r$, and the
relative error caused by this choice of $r$ is at most the
\emph{$r$-Steiner ratio} $\rho_r = 1+\Theta(1/\log
r)$~\cite{BD97}. But this is not so practical: to get relative error
$1+\epsilon$, we apply the ellipsoid algorithm to an LP with
$|R|^{\exp(\Theta(1/\epsilon))}$ variables!

In the \emph{unrestricted setting} where $\K$ contains all possible
full components without regard to size, it is an open problem to
optimize any of the hypergraphic LPs exactly in polynomial time. We
make some progress here: in quasibipartite instances, the proof method
of our hypergraphic-bidirected equivalence theorem
(\prettyref{sec:lifting}) implies that one can exactly compute the LP
optimal value, and a dual optimal solution. Regarding this open
problem, we note that the $r$-restricted LP optimum is at most
$\rho_r$ times the unrestricted optimum, and wonder whether there
might be some advantage gained by using the fact that the hypergraphic
LPs have sparse optima.

We reiterate our feeling that it is important to obtain practical
algorithms and understand the bidirected cut relaxation as well as
possible, e.g.\ we know now that it has an integrality gap of at most
1.28 on quasi-bipartite instances, but obtaining such a bound {\em
  directly} could give new insights.
}{}
\myifthen{
\subsubsection{Other Related Work}
In the special case of $r$-restricted instances for $r=3$, the partition hypergraphic LP is essentially a special case of an LP introduced by Vande Vate~\cite{Vate92} for matroid matching, which is totally dual half-integral~\cite{GP08}. Additional facts about the hypergraphic relaxations appear in the thesis of the third author~\cite{P09thesis}, e.g.~a combinatorial ``gainless tree formulation" for the LPs similar in flavour to the ``1-tree bound" for the Held-Karp TSP relaxation.
}{}
\section{Uncrossing Partitions}
\label{sec:up} \label{sec:setuncrossing}
In this section we are interested in {\em uncrossing} a minimal set of {\em tight partitions}
that uniquely define a basic feasible solution to \eqref{eq:LP-PU}. We start with a few preliminaries necessary to state our result formally.
\subsection{Preliminaries}
We introduce some needed well-known properties of partitions that arise in combinatorial lattice theory~\cite{Stan86}.
\begin{definition}
 We say that a partition $\pi'$ \emph{refines} another partition $\pi$ if each part of $\pi'$ is contained in some
part of $\pi$. We also say $\pi$ {\em coarsens} $\pi'$. Two partitions \emph{cross} if neither refines the other.
A family of partitions forms a \emph{chain} if no pair of them cross. Equivalently, a chain is any family $\pi^1, \pi^2, \dotsc, \pi^t$ such that $\pi^i$ refines $\pi^{i-1}$ for each $1 < i \leq t$.
\end{definition}

The family $\Pi_R$ of all partitions of $R$ forms a \emph{lattice}
with a \emph{meet operator} $\wedge : \Pi_R^2 \to \Pi_R$ and a
\emph{join operator} $\vee : \Pi_R^2 \to \Pi_R$. The meet $\pi \wedge
\pi'$ is the coarsest partition that refines both $\pi$ and $\pi'$,
and the join $\pi \vee \pi'$ is the most refined partition that
coarsens both $\pi$ and $\pi'$. See \prettyref{fig:partitions} for an
illustration.

\begin{definition}[Meet of partitions]\label{definition:meet}Let the parts of $\pi$ be $\pi_1, \dotsc, \pi_t$ and let the
parts of $\pi'$ be $\pi'_1, \dotsc, \pi'_u$. Then the parts of the
meet $\pi \wedge \pi'$ are the nonempty intersections of parts of
$\pi$ with parts of $\pi'$,
$$\pi \wedge \pi' = \{\pi_i \cap \pi'_j \mid 1 \leq i \leq t, 1 \leq
j \leq u \textrm{ and } \pi_i \cap \pi'_j \neq \varnothing\}.$$
\end{definition}

\noindent
Given a graph $G$ and a partition $\pi$ of $V(G)$, we say that $G$
\emph{induces} $\pi$ if the parts of $\pi$ are the vertex sets of
the connected components of $G$.
\begin{definition}[Join of partitions]\label{definition:join}
Let $(R, E)$ be a graph that induces $\pi$, and let $(R, E')$ be a
graph that induces $\pi'$. Then the graph $(R, E \cup E')$ induces
$\pi \vee \pi'$.
\end{definition}

\comment{\begin{pspicture}(-2,-2)(2,2)\psset{unit=0.66}
\terminal{0,1}{t1}\terminal{1,0}{t2}\terminal{-1,0}{t3}\terminal{0,-1}{t4}
\terminal{1,2}{t1a}\terminal{2,1}{t2a}\terminal{-1,2}{t3a}\terminal{-2,1}{t4a}
\terminal{1,-2}{t1b}\terminal{2,-1}{t2b}\terminal{-1,-2}{t3b}\terminal{-2,-1}{t4b}
\psccurve(-1,-0.4)(-2.4,1)(-1,2.4)(0.4,1)
\psccurve(-2.4,-0.8)(-1,-0.5)(0.4,-0.8)(1.4,-2.2)(-1.4,-2.2)
\psccurve(0.8,2.4)(0.5,1)(0.8,-0.4)(2.2,-1.4)(2.2,1.4)
\psset{linestyle=dashed}
\psccurve(-1,-0.2)(-2.2,1)(-1,2.2)(0.2,1)
\pscircle(1,-2){0.3}
\pscircle(2,-1){0.3}
\pscircle(1,2){0.3}
\psccurve(0.8,-0.1)(1.2,0.8)(2.2,1.1)(1.8,0.2)
\psccurve(-2.2,-0.9)(-1,-0.6)(0.2,-0.9)(-1,-2.2)
\rput*(0,-3){(a): The dashed partition refines the solid one.}
\end{pspicture}}

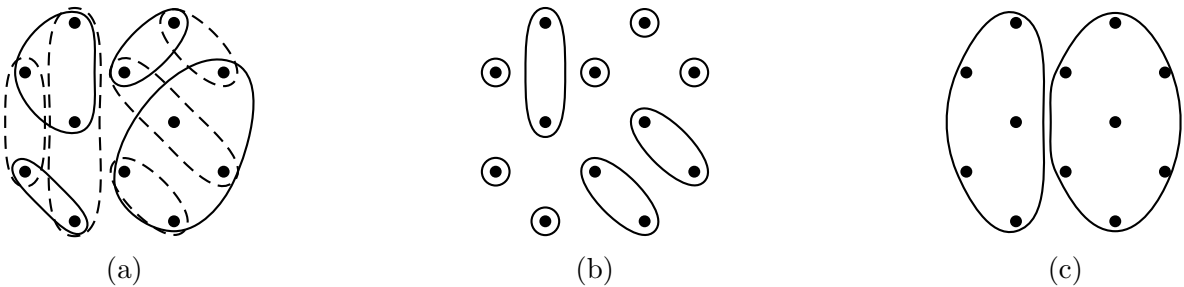
\begin{figure}[htb]
\begin{center} \leavevmode
\begin{pspicture}(-2,-2)(2,2)\psset{unit=0.66}
\terminal{0,1}{t1}\terminal{1,0}{t2}\terminal{-1,0}{t3}\terminal{0,-1}{t4}
\terminal{1,2}{t1a}\terminal{2,1}{t2a}\terminal{-1,2}{t3a}\terminal{-2,1}{t4a}
\terminal{1,-2}{t1b}\terminal{2,-1}{t2b}\terminal{-1,-2}{t3b}\terminal{-2,-1}{t4b}
\psccurve(-2.2,1)(-0.9,-0.2)(-0.6,1)(-0.9,2.2)
\psccurve(-2.2,-0.8)(-1.7,-1.7)(-0.8,-2.2)(-1.3,-1.3)
\psccurve(-0.2,0.8)(0.2,1.8)(1.2,2.2)(0.8,1.2)
\psccurve(-0.2,-1.2)(2.2,1.2)(2.2,-1.2)(1,-2.2)
\psset{linestyle=dashed}
\psccurve(2.2,0.8)(1.8,1.8)(0.8,2.2)(1.2,1.2)
\psccurve(1.2,-2.2)(0.8,-1.2)(-0.2,-0.8)(0.2,-1.8)
\psccurve(-2,-1.3)(-2.4,0)(-2,1.3)(-1.6,0)
\psccurve(-1,-2.3)(-1.5,0)(-1,2.3)(-0.5,0)
\psccurve(-0.2,1.2)(1.3,0.3)(2.2,-1.2)(0.7,-0.3) \rput*(0,-3){(a)}
\end{pspicture}\hfill
\begin{pspicture}(-2,-2)(2,2)\psset{unit=0.66}
\pscircle(-2,-1){0.3}
\pscircle(-1,-2){0.3}
\pscircle(-2,1){0.3}
\pscircle(0,1){0.3}
\pscircle(1,2){0.3}
\pscircle(2,1){0.3}
\psccurve(-1,-0.3)(-1.4,1)(-1,2.3)(-0.6,1)
\psccurve(1.2,-2.2)(0.8,-1.2)(-0.2,-0.8)(0.2,-1.8)
\psccurve(2.2,-1.2)(1.8,-0.2)(0.8,0.2)(1.2,-0.8)
\terminal{0,1}{t1}\terminal{1,0}{t2}\terminal{-1,0}{t3}\terminal{0,-1}{t4}
\terminal{1,2}{t1a}\terminal{2,1}{t2a}\terminal{-1,2}{t3a}\terminal{-2,1}{t4a}
\terminal{1,-2}{t1b}\terminal{2,-1}{t2b}\terminal{-1,-2}{t3b}\terminal{-2,-1}{t4b}
\rput*(0,-3){(b)}
\end{pspicture}\hfill
\begin{pspicture}(-2,-2)(2,2)\psset{unit=0.66}
\psccurve(-2.2,1)(-2.2,-1)(-1,-2.2)(-0.45,0)(-1,2.2)
\psccurve(-0.2,1)(-0.3,0)(-0.2,-1)(1,-2.2)(2.1,-1.1)(2.1,1.1)(1,2.2)
\terminal{0,1}{t1}\terminal{1,0}{t2}\terminal{-1,0}{t3}\terminal{0,-1}{t4}
\terminal{1,2}{t1a}\terminal{2,1}{t2a}\terminal{-1,2}{t3a}\terminal{-2,1}{t4a}
\terminal{1,-2}{t1b}\terminal{2,-1}{t2b}\terminal{-1,-2}{t3b}\terminal{-2,-1}{t4b}
\rput*(0,-3){(c)}
\end{pspicture}
\end{center}
\caption{\small Illustrations of some partitions. The black dots are the
terminal set $R$. (a): two partitions; neither refines the other. (b): the meet of the partitions from (a). (c): the join of the partitions from (a).} \label{fig:partitions}
\end{figure}
\def\meet{\wedge}
\def\join{\vee}
\noindent

Given a feasible solution $x$ to \eqref{eq:LP-PU}, a partition $\pi$
is \emph{tight} if $\sum_{K\in \K}x_K\rc^\pi_K = r(\pi) - 1$.  Let
$\tight(x)$ be the set of all tight partitions. We are interested in
{\em uncrossing} this set of partitions. More precisely, we wish to
find a cross-free set of partitions (chain) which uniquely defines
$x$.  One way would be to prove the following.

\begin{ppty}\label{ppty:quote}
  If two crossing partitions $\pi$ and $\pi'$ are in $\tight(x)$, then
  so are $\pi \meet \pi'$ and $\pi \join \pi'$.
\end{ppty}

This type of property is already well-used~\myifthen{\cite{CNF85,EG77,J98,SL07}}{\cite{EG77,J98,SL07}}
for sets (with meets and joins replaced by unions and intersections
respectively), and the standard approach is the following. The typical proof
considers the constraints in \eqref{eq:LP-PU} corresponding to $\pi$
and $\pi'$ and uses the ``supermodularity'' of the RHS and the
``submodularity'' of the coefficients in the LHS. In particular, if
the following is true,
\begin{align}
  \forall \pi, \pi':~ r(\pi \vee \pi')+r(\pi \wedge \pi') &~~~\geq~~~ r(\pi) + r(\pi') \label{eq:naive1} \\
  \forall K, \pi, \pi': ~ \rc_K^\pi + \rc_K^{\pi'} &~~~\ge~~~
  \rc_K^{\pi \vee \pi'} + \rc_K^{\pi \wedge \pi'} \label{eq:naive}
\end{align}
then \prettyref{ppty:quote} can be proved easily by writing a
string of inequalities.\footnote{In this hypothetical scenario we get
  $r(\pi) + r(\pi') - 2 = \sum_{K} x_K(\rc_K^{\pi}+\rc_K^{\pi'}) \ge
  \sum_{K} x_K(\rc_K^{\pi\meet\pi'}+\rc_K^{\pi\join\pi'}) \ge
  r(\pi\meet\pi')+r(\pi\join\pi') - 2 \ge r(\pi)+r(\pi')-2$; thus the
  inequalities hold with equality, and the middle one shows $\pi
  \wedge \pi'$ and $\pi \vee \pi'$ are tight.}

Inequality \eqref{eq:naive1} is indeed true (see, for example,
\cite{Stan86}), but unfortunately inequality \eqref{eq:naive} is not
true in general, as the following example shows.
\begin{example}\label{example:nonsubm}Let $R = \{1, 2, 3, 4\}$, $\pi =
  \{\{1, 2\}, \{3, 4\}\}$ and $\pi' = \{\{1, 3\}, \{2, 4\}\}.$ Let $K$
  denote the full component $\{1, 2, 3, 4\}$. Then $\rc_K^\pi +
  \rc_K^{\pi'} = 1 + 1 < 0 + 3 = \rc_K^{\pi \vee \pi'} + \rc_K^{\pi
    \wedge \pi'}.$
\end{example}

Nevertheless, \prettyref{ppty:quote} is true; its correct proof is
given in \myifthen{\prettyref{sec:pui}}{the full version of this paper
  \cite{CKP09}} and depends on a simple though subtle extension of the
usual approach. The crux of the insight needed to fix the approach is
not to consider \emph{pairs} of constraints in \eqref{eq:LP-PU}, but
rather multi-sets which may contain more than two inequalities. Using
this uncrossing result, we can prove the following theorem (details
are given in \myifthen{\prettyref{sec:pbu}}{\cite{CKP09}}). Here, we let
$\underline{\pi}$ denote $\{R\}$, the unique partition with (minimal)
rank 1; later we use $\overline{\pi}$ to denote $\{\{r\}\mid r\in
R\}$, the unique partition with (maximal) rank $|R|$.

\begin{theorem}\label{thm:1}
  Let $x^*$ be a basic feasible solution of \eqref{eq:LP-PU}, and let
  \C\ be an inclusion-wise maximal chain in $\tight(x^*) \bs \underline{\pi}$. Then
  $x^*$ is uniquely defined
  by
  \begin{equation} \label{eq:chain}
    \sum_{K \in \K} \rc_K^\pi x^*_K = r(\pi)-1 \quad \forall \pi \in \C.
  \end{equation}
\end{theorem}

Any chain of distinct partitions of $R$ that does not contain
$\underline{\pi}$ has size at most $|R|-1$, and this is an upper bound
on the rank of the system in \eqref{eq:chain}. Elementary linear
programming theory immediately yields the following corollary.

\begin{corollary}\label{corollary:structure}
  Any basic solution $x^*$ of \eqref{eq:LP-PU} has at most $|R|-1$ non-zero coordinates.
\end{corollary}

\myifthen{
\subsection{Partition Uncrossing Inequalities}\label{sec:pui}
We start with the following definition.
\begin{definition}\label{definition:merged}
  Let $\pi \in \Pi_R$ be a partition and let $S \subset R$. Define the
  \emph{merged partition} $m(\pi, S)$ to be the most refined partition
  that coarsens $\pi$ and contains all of $S$ in a single part. See
  \prettyref{fig:merged} for an example.  Informally, $m(\pi,S)$ is
  obtained by merging all parts of $\pi$ which intersect $S$.
  Formally, $m(\pi,S)$ equals the set of parts $\{\{\pi_j\}_{j:
    \pi_j\cap S=\varnothing}, \bigcup_{j:\pi_j\cap S\neq \varnothing}
  \pi_j\}$.
\end{definition}

\begin{figure}[htb]
\myifthen{}{\psset{unit=0.8cm}}
\begin{center} \leavevmode
\begin{pspicture}(-3,-2.5)(3,2.5)
\terminal{0,1}{t1}\terminal{1,0}{t2}\terminal{-1,0}{t3}\terminal{0,-1}{t4}
\terminal{1,2}{t1a}\terminal{2,1}{t2a}\terminal{-1,2}{t3a}\terminal{-2,1}{t4a}
\terminal{1,-2}{t1b}\terminal{2,-1}{t2b}\terminal{-1,-2}{t3b}\terminal{-2,-1}{t4b}
\psccurve(2.2,0.8)(1.8,1.8)(0.8,2.2)(1.2,1.2)
\psccurve(1.2,-2.2)(0.8,-1.2)(-0.2,-0.8)(0.2,-1.8)
\psccurve(-2,-1.3)(-2.4,0)(-2,1.3)(-1.6,0)
\psccurve(-1,-2.3)(-1.5,0)(-1,2.3)(-0.5,0)
\psccurve(-0.2,1.2)(1.3,0.3)(2.2,-1.2)(0.7,-0.3)
\psccurve[linestyle=dashed](-1.2,0.1)(1.2,0.2)(2.2,-1)(1,-2.2)
\end{pspicture}
\hspace{0.8cm}
\begin{pspicture}(-3,-2.5)(3,2.5)
\terminal{0,1}{t1}\terminal{1,0}{t2}\terminal{-1,0}{t3}\terminal{0,-1}{t4}
\terminal{1,2}{t1a}\terminal{2,1}{t2a}\terminal{-1,2}{t3a}\terminal{-2,1}{t4a}
\terminal{1,-2}{t1b}\terminal{2,-1}{t2b}\terminal{-1,-2}{t3b}\terminal{-2,-1}{t4b}
\psccurve(2.2,0.8)(1.8,1.8)(0.8,2.2)(1.2,1.2)
\psccurve(-2,-1.3)(-2.4,0)(-2,1.3)(-1.6,0)
\psccurve(-1,-2.3)(-1.4,0)(-1,2.3)(0.8,0.8)(2.2,-1.2)(1,-2.2)
\end{pspicture}
\end{center}
\caption{Illustration of merging. The left figure shows a (solid)
partition $\pi$ along with a (dashed) set $S$. The right figure
shows the merged partition $m(\pi, S)$.} \label{fig:merged}
\end{figure}
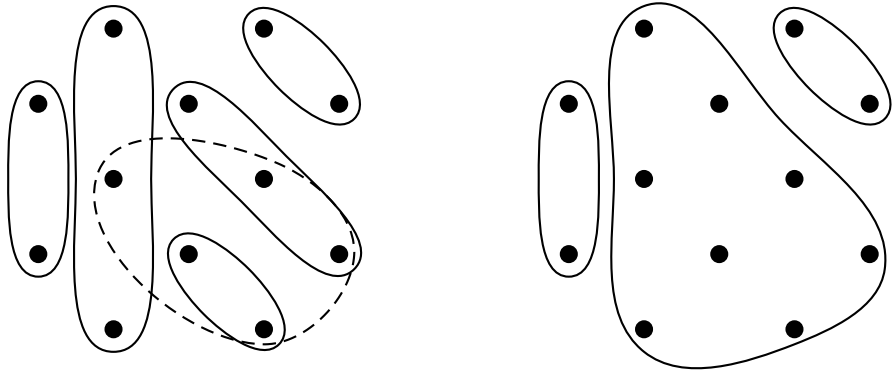
\noindent
We will use the following straightforward fact later:
\begin{equation}\label{eq:rdrop}\rc_K^\pi = r(\pi)-r(m(\pi,
K)).\end{equation}

We now state the (true) inequalities which replace the false inequality \eqref{eq:naive}. Later, we show how one uses these to obtain
partition uncrossing, e.g.\ to prove \prettyref{ppty:quote}.

\begin{lemma}[Partition Uncrossing Inequalities]\label{lemma:pu}
Let $\pi, \pi' \in \Pi_R$ and let the parts of $\pi$ be $\pi_1,
\pi_2, \dotsc, \pi_{r(\pi)}$.
\begin{eqnarray}
\label{eq:const} r(\pi) \left[r(\pi')-1\right] +
\left[r(\pi)-1\right] &=& \left[r(\pi \wedge \pi')-1\right] +
\sum_{i=1}^{r(\pi)} \left[r(m(\pi', \pi_i))-1\right]\\
\label{eq:coeff}
\forall K \in \K: \quad r(\pi) \Bigl[\rc_K^{\pi'}\Bigr] +
\Bigl[\rc_K^\pi\Bigr] &\geq& \Bigl[\rc_K^{\pi \wedge \pi'}\Bigr] +
\sum_{i=1}^{r(\pi)} \Bigl[\rc_K^{m(\pi', \pi_i)}\Bigr]
\end{eqnarray}
\end{lemma}
\noindent
Before giving the proof of the above lemma, let us first show how it can be used to prove the statement \prettyref{ppty:quote}. \\

\noindent
{\em Proof of \prettyref{ppty:quote}}.
Since $\pi$ and $\pi'$ are tight,
\begin{align*}
r(\pi)[r(\pi') - 1] + [r(\pi) - 1] = r(\pi)\Bigl[\sum_{K} x_K\rc_K^{\pi'}\Bigr] + \Bigl[\sum_K x_K\rc_K^\pi\Bigr] = \sum_K x_K \biggl(r(\pi)\Bigl[\rc_K^{\pi'}\Bigr] + \Bigl[\rc_K^\pi\Bigr]\biggr) \\
 \ge \sum_K x_K\biggl(\Bigl[\rc_K^{\pi \wedge \pi'}\Bigr] +
\sum_{i=1}^{r(\pi)} \Bigl[\rc_K^{m(\pi', \pi_i)}\Bigr]\biggr) = \sum_K x_K\Bigl[\rc_K^{\pi \wedge \pi'}\Bigr] + \sum_{i=1}^{r(\pi)} \sum_K x_K \Bigl[\rc_K^{m(\pi', \pi_i)}\Bigr] \\
\ge \left[r(\pi \wedge \pi')-1\right] +
\sum_{i=1}^{r(\pi)} \left[r(m(\pi', \pi_i))-1\right] = r(\pi) \left[r(\pi')-1\right] +
\left[r(\pi)-1\right]
\end{align*}
where the first inequality follows from \eqref{eq:coeff} and the second from~\eqref{eq:LP-PU2} (as $x$ is feasible); the last equality is \eqref{eq:const}. Since the first and last terms are equal, all the inequalities are equalities, in particular our application of \eqref{eq:LP-PU2} shows that $\pi \wedge \pi'$ and each $m(\pi',\pi_i)$ is tight.
Iterating the latter fact, we see that $m(\dotsb m(m(\pi', \pi_1), \pi_2), \dotsb)=\pi \vee \pi'$ is also tight. $\square$\\

\noindent
To prove the inequalities in \prettyref{lemma:pu}  we need the following lemma that relates the rank of sets and the rank contribution of partitions. Recall $\rho(X) := \max(0,|X|-1)$.
\begin{lemma} \label{lemma:rc} For a partition $\pi = \{\pi_1, \dotsc, \pi_t\}$ of $R$, where $t = r(\pi)$, and for any $K\subseteq R$, we have
$$\rho(K) = \rc_K^\pi + \sum_{i=1}^t\rho(K \cap \pi_i).$$
\end{lemma}
\begin{proof}
By definition, $K \cap \pi_i \neq \varnothing$ for exactly $1 +
\rc_K^\pi$ values of $i$. Also, $\rho(K \cap \pi_i)=0$ for all other
$i$. Hence \begin{equation}\sum_{i=1}^t\rho(K \cap \pi_i) =
\sum_{i:K \cap \pi_i \neq \varnothing} (|K \cap \pi_i|-1) =
\left(\sum_{i:K \cap \pi_i \neq \varnothing} |K \cap
\pi_i|\right)-(\rc_K^\pi+1).\label{eq:houston}\end{equation} Observe
that $ \sum_{i:K \cap \pi_i \neq \varnothing} |K \cap
\pi_i|=|K|=\rho(K)+1$; using this fact together with Equation
\eqref{eq:houston} we obtain
$$\sum_{i=1}^t\rho(K \cap \pi_i) =
\left(\sum_{i:K \cap \pi_i \neq \varnothing} |K \cap
\pi_i|\right)-(\rc_K^\pi+1) = \rho(K)-1+(\rc_K^\pi+1).$$
Rearranging, the proof of \prettyref{lemma:rc} is complete.
\end{proof}

\noindent
{\em Proof of \prettyref{lemma:pu}}.
First, we argue that $\pi \wedge \pi' = \overline{\pi}$ holds without loss of generality.
In the general case, for each part $p$ of $\pi \wedge \pi'$ with
$|p|\ge 2$, contract $p$ into one pseudo-vertex and define
the new $K$ to include the pseudo-vertex corresponding to $p$ if and
only if $K \cap p \neq \varnothing$. This contraction does not affect
the value of any of the terms in Equations \eqref{eq:coeff} and
\eqref{eq:const}, so is without loss of generality. After contraction, for any part $\pi_i$ of
$\pi$ and part $\pi'_j$ of $\pi'$, we have $|\pi_i \cap \pi'_j| \leq
1$, so indeed $\pi \wedge \pi' = \overline{\pi}$.

\begin{proof}[Proof of Equation \eqref{eq:const}]
Fix $i$. Since $|\pi_i \cap \pi'_j| \leq 1$ for all $j$, the rank
contribution $\rc_{\pi_i}^{\pi'}$ is equal to $|\pi_i|-1.$ Then
using Equation \eqref{eq:rdrop} we know that $r(m(\pi', \pi_i)) =
r(\pi') - |\pi_i|+1$. Thus adding over all $i$, the right-hand side
of Equation \eqref{eq:const} is equal to
$$|R|-1 + \sum_{i=1}^{r(\pi)}(r(\pi')-|\pi_i|) = |R|-1 + r(\pi)r(\pi')-|R|$$
and this is precisely the left-hand side of Equation
\eqref{eq:const}.
\end{proof}

\begin{proof}[Proof of Equation \eqref{eq:coeff}]
Fix $i$. Since $|\pi_i \cap \pi'_j| \leq 1$ for all $j$, we have
\begin{equation}
\rc_K^{\pi'} - \rc_K^{m(\pi', \pi_i)} \geq \rho(\pi_i \cap
K)\label{eq:tight}
\end{equation}
because, when we merge the parts of $\pi'$ intersecting $\pi_i$, we
make $K$ span at least $\rho(\pi_i \cap K)$ fewer parts.  Note that the inequality could be strict if both $\pi_i$ and $K$ intersect a part of $\pi'$ without having a common vertex in that part.

Adding the right-hand side of Equation \eqref{eq:tight} over all $i$
gives
\begin{equation}\sum_{i=1}^{r(\pi)}(\rc_K^{\pi'} - \rc_K^{m(\pi', \pi_i)}) \geq \sum_{i=1}^{r(\pi)}\rho(\pi_i \cap K) = \rho(K)-\rc_K^\pi.\label{eq:snore}\end{equation}
where the last equality follows from \prettyref{lemma:rc}. To finish
the proof we observe $\rho(K) = \rc_K^{\pi\meet\pi'}$, since $\pi
\meet \pi' = \overline{\pi}$.
\end{proof}
\noindent
This completes the proof of \prettyref{lemma:pu}. $\hfill \Box$\\

\subsection{Sparsity of Basic Feasible Solutions: Proof of \prettyref{thm:1}}\label{sec:pbu}
\def\row{{\tt row}}
\begin{proof}
  Let $\supp(x^*)$ be the full components $K$ with $x^*_K>0$. Consider
  the constraint submatrix with rows corresponding to the
  tight partitions and columns corresponding to the full components in
  $\supp(x^*)$. Since $x^*$ is a basic feasible solution, any full-rank
  subset of rows uniquely defines $x^*$. We now show that any maximal chain $\C$ in $\tight(x^*)$
  corresponds to such a subset.

Let $\row(\pi) \in \R^{\supp(x^*)}$ denote the row corresponding to partition $\pi$ of
this matrix, i.e., $\row(\pi)_K = \rc^\pi_K$, and given a collection $\mathcal{R}$ of partitions (rows), let
$\spn(\mathcal{R})$ denote the linear span of the rows in $\mathcal{R}$.
We now prove that for any tight partition $\pi \notin \C$, we have
$\row(\pi) \in \spn(\C)$; this will complete the proof of the theorem.

For sake of contradiction, suppose $\row(\pi) \not\in \spn(\C)$. Choose $\pi$ to be the
counterexample partition with smallest rank $r(\pi)$. Firstly, since
$\C$ is maximal, $\pi$ must cross some partition $\sigma$ in $\C$.
Choose $\sigma$ to be the most refined partition in $\C$ which crosses
$\pi$. Let the parts of $\sigma$ be $(\sigma_1,\ldots,\sigma_t)$. The
following claim uses the partition uncrossing inequalities to derive a
linear dependence between the rows corresponding to $\sigma,\pi$ and
the partitions formed by merging parts of $\sigma$ with $\pi$.

\begin{claim}\label{claim:talon}
We have $\row(\sigma) + |r(\sigma)|\cdot \row(\pi) = \row(\pi\meet\sigma) + \sum_{i=1}^t \row(m(\pi,\sigma_i))$.
\end{claim}
\begin{proof}
  Since $\sigma$ and $\pi$ are both tight partitions, the proof of
  \prettyref{ppty:quote} shows that the partition inequality
  \eqref{eq:coeff} holds with equality for all $K\in \supp(x^*)$,
  $\pi$ and $\sigma$, implying the claim.
\end{proof}

Let $\cp_\pi(\sigma)$ be the parts of $\sigma$ which intersect at
least two parts of $\pi$; i.e., merging the parts of $\pi$ that
intersect $\sigma_i$, for any $\sigma_i \in \cp_\pi(\sigma)$,
decreases the rank of $\pi$. Formally,
$$\cp_\pi(\sigma) := \{ \sigma_i \in \sigma: ~~ m(\pi,\sigma_i) \neq \pi \}$$
Note that one can modify \prettyref{claim:talon} by subtracting $(r(\sigma) - |\cp_\pi(\sigma)|)\row(\pi)$ from both sides to get
\begin{equation}\label{eq:lincomb}
\row(\sigma) + |\cp_\pi(\sigma)|\cdot \row(\pi) = \row(\pi\meet\sigma) + \sum_{\sigma_i\in \cp_\pi(\sigma)} \row(m(\pi,\sigma_i))
\end{equation}

Now if $\row(\pi) \notin \spn(\C)$, we must have either $\row(\pi\meet\sigma)$ is not in $\spn(\C)$ or $\row(m(\pi,\sigma_i))$ is not in $\spn(\C)$ for some $i$.
We show that either case leads to the needed contradiction, which will prove the theorem.
\begin{description}
\item[Case 1:]$\row(\pi \meet \sigma) \notin \spn(\C)$. Note there is $\sigma'\in \C$ which crosses $\pi\meet\sigma$, since $\pi\meet\sigma$ is not in the maximal chain $\C$.
Since $\sigma', \sigma \in \C$ and by considering the refinement order, it is easy to see that $\sigma'$ (strictly) refines $\sigma$ and $\sigma'$ crosses $\pi$. This contradicts our choice of $\sigma$ as the most refined partition in $\C$ crossing $\pi$, since $\sigma'$ was also a candidate.
\item[Case 2:]
$\row(m(\pi, \sigma_i)) \not\in \spn(\C)$. Note $m(\pi, \sigma_i)$ is also tight.
Since $\sigma_i \in \cp_\pi(\sigma)$, $m(\pi, \sigma_i)$ has smaller rank than
$\pi$. This contradicts our choice of $\pi$.
\end{description}
This completes the proof of \prettyref{thm:1}.
\end{proof}}{}

\section{Equivalence of Formulations}\label{sec:equiv}

In this section we describe our equivalence results. A summary of the
known and new results is given in \prettyref{fig:overview}.

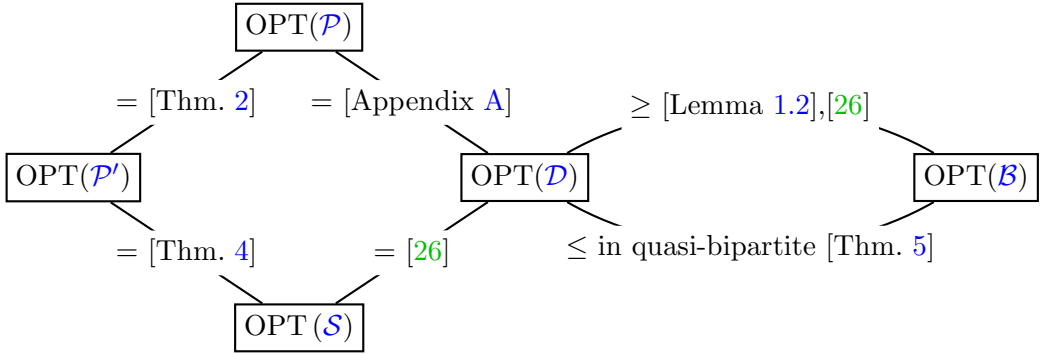
\begin{figure}[htb]
  \begin{center} \leavevmode
    \myifthen{}{\psset{unit=0.8cm}}
    \begin{pspicture}(-2,-2)(14,2.5)

    %\cnodeput(3,-3){S}{\OPT\eqref{eq:LP-S}}
    \rput(3,-2){\rnode{S}{\psframebox{\ensuremath{\OPT\eqref{eq:LP-S}}}}}
    \rput(3,2){\rnode{PU}{\psframebox{\OPT\eqref{eq:LP-PU}}}}
    \rput(0,0){\rnode{P}{\psframebox{\OPT\bdp}}}
    \rput(6,0){\rnode{D}{\psframebox{\OPT\eqref{eq:LP-PUDir}}}}
    \rput(12,0){\rnode{B}{\psframebox{\OPT\eqref{eq:LP-B}}}}

\psset{arcangle=30}
\ncline{-}{PU}{P}\mput*{$=$ {[Thm.~\ref{theorem:p-pu}]}}
\ncline{-}{P}{S}\mput*{$=$ {[Thm.~\ref{theorem:spe}]}}
\ncline{-}{S}{D}\mput*{$=$ {\cite{PVd03}}}
\ncline{-}{PU}{D}\mput*{$=$ {\myifthen{[Appendix \ref{app:reproof}]}{\cite{CKP09}}}}
\ncarc{-}{D}{B}\mput*{$\ge$ {[Lemma \ref{lem:simple}],\cite{PVd03}}}
\ncarc{-}{B}{D}\mput*{$\le$ in quasi-bipartite {[Thm.~\ref{theorem:lifting}]}}

    \end{pspicture}
  \end{center}
  \caption{Summary of relations among various LP relaxations}
  \label{fig:overview} \end{figure}

\myifthen{
As we mentioned in the introduction, we give a redundant set of proofs
for completeness and to demonstrate novel techniques. The proof that \eqref{eq:LP-PU} and \eqref{eq:LP-PUDir} have the same value, which appears in Appendix \ref{app:reproof}, is a consequence of hypergraph orientation results of Frank et al.~\cite{FKK03b}.

\subsection{Bounded and Unbounded Partition Relaxations}

\begin{theorem}\label{theorem:p-pu}
  The LPs \eqref{eq:LP-P2} and \eqref{eq:LP-PU} have the same optimal
  value.
\end{theorem}

\noindent We actually prove a stronger statement.
\begin{definition}The collection $\K$ of hyperedges is
  \emph{down-closed} if whenever $S \in \K$ and $\varnothing \neq T
  \subset S$, then $T \in \K.$ For down-closed $\K$, the cost function
  $C: \K \to \R_+$ is \emph{non-decreasing} if $C_S \leq C_T$ whenever
  $S \subset T$.\end{definition}

\begin{theorem}\label{theorem:p-pu-app}
  If the set of hyperedges is down-closed and the cost function is
  non-decreasing, then \eqref{eq:LP-P2} and \eqref{eq:LP-PU} have the
  same optimal value.
\end{theorem}
\prettyref{theorem:p-pu-app} implies \prettyref{theorem:p-pu} since
the hypergraph and cost function derived from instances of the Steiner
tree problem are down-closed and non-decreasing (e.g.~$C_{\{k\}} = 0$
for every $k \in R$; we remark that the variables $x_{\{k\}}$ act just
as placeholders).  Our proof of \prettyref{theorem:p-pu} relies on the
following operation which we call {\em shrinking}.

\begin{definition}
Given an assignment  $x:\K \to \R_+$ to the full components, suppose $x_K > 0$ for some $K$. The operation $\Shrink(x,K,K',\delta)$, where $K' \subseteq K$, $|K'| = |K|-1$ and  $0 < \delta \le x_K$, changes $x$ to $x'$ by decreasing $x'_K := x_K - \delta$ and increasing $x'_{K'}  := x_{K'} + \delta$. \end{definition}
% We call the parameter $\delta$ the {\em shrinkage} of the shrink operation. -- never used
\noindent
Note that shrinking is defined only for down-closed hypergraphs. Also note that on performing a shrinking operation, the cost of the solution cannot increase, if the cost function is non-decreasing. The theorem is proved by taking the optimum solution to \eqref{eq:LP-PU} which  minimizes the sum $\sum_{K\in \K} x_K|K|$, and then showing that this must satisfy the equality in \eqref{eq:LP-P2}, or a shrinking operation can be performed. Now we give the details.

\begin{proof}[Proof of \prettyref{theorem:p-pu-app}]
It suffices to exhibit an optimum solution of \eqref{eq:LP-PU} which satisfies the equality in \eqref{eq:LP-P2}.
Let $x$  be an optimal solution to \eqref{eq:LP-PU} which minimizes the sum $\sum_{K\in \K} x_K|K|$.

\begin{claim}\label{claim:inter}
For every $K$ with $x_K>0$ and for every $r\in K$, there exists a tight partition (w.r.t.\ $x$) $\pi$ such that the part of $\pi$ containing
$r$ contains no other vertex of $K$.
\end{claim}
\begin{proof}
Let $K'=K\setminus\{r\}$. If the above is not true, then this implies that for every tight partition $\pi$, we have $\rc_K^\pi = \rc_{K'}^\pi$.
We now claim that there is a $\delta > 0$ such that we can perform $\Shrink(x,K,K',\delta)$ while retaining feasibility in \eqref{eq:LP-PU}. This is a contradiction since the shrink operation
strictly reduces $\sum_K |K|x_K$ and doesn't increase cost.
Specifically, take
\begin{center}
$\delta := \min\{x_K, \min_{\pi: \rc_{K'}^\pi \neq \rc_K^\pi}\sum_K \rc_K^\pi x_K - r(\pi)+1\}$
\end{center}
which is positive since for tight partitions we have $\rc_K^\pi = \rc_{K'}^\pi$.
\end{proof}

Let $\tight(x)$ be the set of tight partitions, and $\pi^* :=
\bigwedge \{ \pi \mid \pi \in \tight(x)\}$ the meet of all tight
partitions. By \prettyref{ppty:quote}, $\pi^*$ is tight. By
\prettyref{claim:inter}, for any $K$ with $x_K > 0$, we have
$\rc_K^{\pi^*} = |K|-1$. Thus, $r(\pi^*)-1 = \sum_{K\in \K}
x_K\rc_K^{\pi^*} = \sum_{K\in\K} x_K(|K| - 1) \ge r(\overline{\pi}) -
1$. But since $\overline{\pi}$ is the unique maximal-rank partition,
this implies $\pi^* = \overline{\pi}$. Thus $\overline{\pi}$ is
tight. This implies $x \in \eqref{eq:LP-P2}$.
\end{proof}

\subsection{Partition and Subtour Elimination Relaxations}\label{sec:partsub}

\begin{theorem}\label{theorem:spe}
  The feasible regions of \eqref{eq:LP-P2} and \eqref{eq:LP-S} are the
  same.
\end{theorem}
\begin{proof}
Let $x$ be any feasible solution
  to the LP \eqref{eq:LP-S}. Note that the equality constraint of \eqref{eq:LP-P2}
 is the same as that of \eqref{eq:LP-S}.
 We now show that $x$ satisfies \eqref{eq:LP-PU2}.
  Fix a partition $\pi=\{\pi_1, \dotsc, \pi_t\}$, so $t=r(\pi)$. For
  each $1 \leq i \leq t$, subtract the inequality constraint in \eqref{eq:LP-S}
  with $S = \pi_i$, from the equality constraint in \eqref{eq:LP-S} to obtain
  \begin{equation}
  \sum_{K \in \K} x_K \Bigl( \rho(K)-\sum_{i=1}^t\rho(K \cap \pi_i) \Bigr)
  \geq \rho(R)-\sum_{i=1}^t\rho(\pi_i). \label{eq:derived}\end{equation}
From \prettyref{lemma:rc}, $\rho(K)-\sum_{i=1}^t \rho(K \cap \pi_i) =
\rc_K^\pi$. We also have $\rho(R)-\sum_{i=1}^t\rho(\pi_i) =
|R|-1-(|R|-r(\pi)) = r(\pi)-1$. Thus $x$ is a feasible solution to the LP \eqref{eq:LP-P2}. \\

\noindent
Now, let $x$ be a feasible solution to \eqref{eq:LP-P2} and it suffices to show that it
satisfies the inequality constraints  of \eqref{eq:LP-S}. Fix a set $S\subset R$.
Note when $S = \varnothing$ that inequality constraint is vacuously true so we may assume $S \neq
\varnothing$. Let $R \bs S = \{r_1, \dotsc, r_u\}$. Consider the partition $\pi =
\{\{r_1\}, \dotsc, \{r_u\}, S\}$. Subtract \eqref{eq:LP-PU2} for this
$\pi$ from the equality constraint in \eqref{eq:LP-P2}, to obtain

\begin{equation}\label{eq:donut} \sum_{K \in \K} x_K
(\rho(K)-\rc_K^{\pi}) \leq \rho(R)-r(\pi)+1.\end{equation} Using
\prettyref{lemma:rc} and the fact that $\rho(K \cap \{r_j\}) = 0$ (the set is either empty or a singleton), we get
 $\rho(K)-\rc_K^{\pi} = \rho(K \cap S)$. Finally, as $\rho(R)-r(\pi)+1 = |R|-1-(|R\bs S|+1)+1 = \rho(S),$ the inequality \eqref{eq:donut} is the same as the constraint needed.
 Thus $x$ is a feasible solution to \eqref{eq:LP-S}, proving the theorem.
 \end{proof}

\subsection{Partition and Bidirected Cut Relaxations in Quasibipartite Instances}\label{sec:lifting}

\begin{theorem}\label{theorem:lifting}
On quasibipartite Steiner tree instances, $\OPT\eqref{eq:LP-B} \ge \OPT\eqref{eq:LP-PUDir}$.
\end{theorem}
\noindent

To prove \prettyref{theorem:lifting}, we
look at the duals of the two LPs and we show $\OPT\eqref{eq:LP-BD}
\ge \OPT\eqref{eq:LP-A}$ in quasibipartite instances.
Recall that the support of a solution to \eqref{eq:LP-A} is the family
of sets with positive $z_U$. A family of sets is called \emph{laminar} if for
any two of its sets $A,B$ we have $A\subseteq B, B\subseteq A$, or $A\cap B=\varnothing$.

\begin{lemma} \label{lemma:3lps}
There exists an optimal solution to \eqref{eq:LP-A} whose support is a laminar family of sets.
\end{lemma}
\begin{proof}
  Choose an optimal solution $z$ to \eqref{eq:LP-A} which maximizes
  $\sum_U z_U|U|^2$ among all optimal solutions. We claim that the
  support of this solution is laminar. Suppose not and there exists
  $U$ and $U'$ with $U\cap U' \neq \varnothing$ and $z_U > 0$ and
  $z_{U'} > 0$. Define $z'$ to be the same as $z$
  except $z'_{U} = z_U - \delta$, $z'_{U'} = z_{U'} - \delta$,
  $z'_{U\cup U'} = z_{U\cup U'} + \delta$ and $z'_{U\cap U'} =
  z_{U\cap U'} + \delta$; we will show for small $\delta > 0$, $z'$ is feasible. Note that
  $U\cap U'$ is not empty and $U\cup U'$ doesn't contain $r$, and the
  objective value remains unchanged. Also note that for any $K$ and
  $i\in K$, if $z_{U\cup U'}$ or $z_{U\cap U'}$ appears in the summand of a constraint,
  then at least one of $z_{U}$ or $z_{U'}$ also appears. If both
  $z_{U\cup U'}$ and $z_{U\cap U'}$ appears, then both $z_U$ and
  $z_{U'}$ appears. Thus $z'$ is an optimal solution and $\sum_U
  z'_U|U|^2 > \sum_U z_U|U|^2$, contradicting the choice of $z$.
\end{proof}

\begin{lemma}\label{lem:main-qb}
  For quasibipartite instances, given a solution of \eqref{eq:LP-A}
  with laminar support, we can get a feasible solution to
  \eqref{eq:LP-BD} of the same value.
\end{lemma}
\begin{proof}
  This lemma is the heart of the theorem, and is a little technical to
  prove. We first give a sketch of how we convert a feasible
  solution $z$ of \eqref{eq:LP-A} into a feasible solution to
  \eqref{eq:LP-BD} of the same value.

  Comparing \eqref{eq:LP-A} and \eqref{eq:LP-BD} one first notes that
  the former has a variable for every valid subset of the terminals,
  while the latter assigns values to all valid subsets of the entire
  vertex set.
  %For clarity, we will henceforth use $U$ for valid subset
  % of all vertices, and we use $S$ to denote valid subset of terminals.
  We say that an edge $uv$ is \emph{satisfied} for a candidate
  solution $z$, if both a) $\sum_{U:u\in U, v\notin U} z_{U} \le
  c_{uv}$ and b) $\sum_{U:v\in U, u\notin U} z_{U} \le c_{uv}$
  hold; $z$ is then feasible for \eqref{eq:LP-BD} if {\em all} edges
  are satisfied.

  Let $z$ be a feasible solution to \eqref{eq:LP-A}.
  One easily verifies that all terminal-terminal edges are
  satisfied. On the other hand, terminal-Steiner edges may
  initially not be satisfied. To see this consider the Steiner vertex
  $v$ and its neighbours depicted in \prettyref{fig:lift} below.
  Initially, none of the sets in $z$'s support contains $v$, and
  the load on the edges incident to $v$ is quite {\em skewed}:
  the left-hand side of condition a) above may be large, while
  the left-hand side of condition b) is initially $0$.

  To construct a valid solution for \eqref{eq:LP-BD}, we therefore
  {\em lift} the initial value $z_S$ of each terminal subset $S$ to
  supersets of $S$, by adding Steiner vertices. The lifting
  procedure processes each Steiner vertex $v$ one at a time; when
  processing $v$, we change $z$ by moving dual from some sets $U$ to
  $U \cup \{v\}$. Such a dual transfer decreases the left-hand side
  of condition a) for edge $uv$, and increases the
  (initially $0$) left-hand sides of condition b) for edges connecting $v$ to
  neighbours other than $v$.

  We will soon see that there is a way of carefully lifting duals
  around $v$ that ensures that all edges incident to $v$ become
  satisfied. The definition of our procedure will ensure
  that these edges remain satisfied for the rest of the lifting
  procedure. Since there are no Steiner-Steiner edges, all edges will
  be satisfied once all Steiner vertices are processed.

 \piccaptioninside
  \piccaption{\label{fig:lift} Lifting variable $z_U$.}
  \parpic(5.5cm,4.5cm)[fr]{\includegraphics[scale=.8]{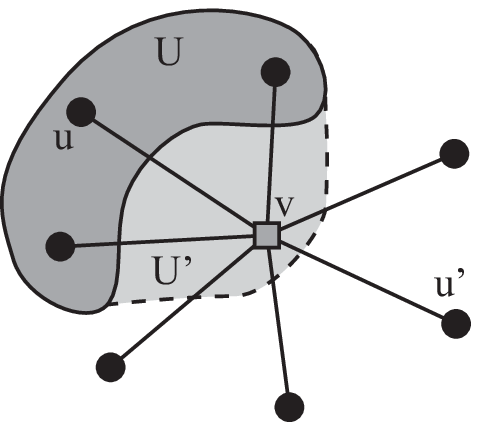}}

  Throughout the lifting procedure, we will maintain that $z$ remains
  unchanged, when projected to the terminals. Formally, we maintain
  the following crucial {\em projection invariant}:
  \begin{equation}\label{eq:pi}\tag{PI}
  \mbox{\begin{minipage}{7cm}
    The quantity
    $\sum_{U: S\subseteq U
      \subseteq S\cup (V\setminus R)} z_{U} $\\[1ex]
    remains constant, for all terminal
    sets $S$.
  \end{minipage}}
\end{equation}
 This invariant leads to two observations: first, the
  constraint \eqref{eq:LP-A2} is satisfied by $z$ at all times, even
  when it is defined on subsets of all vertices; second,
  $\sum_{U\subseteq V} z_U$ is constant throughout, and
  the objective value of $z$ in \eqref{eq:LP-BD} is not affected
  by the lifting. The existence of a lifting of duals around Steiner vertex $v$ such
  that \eqref{eq:pi} is maintained, and such that all edges incident
  to $v$ are satisfied can be phrased as a feasibility problem for
  a linear system of inequalities. We will use
  Farkas' lemma and the feasibility of $z$ for \eqref{eq:LP-A2}
  to complete the proof.

  We now fill in the proof details. Let $\Gamma(v)$ denote the
  set of neighbours of vertex $v$ in the given graph $G$. In each
  iteration, where we process Steiner node $v$, let
 $$ \U_v := \{U: z_U > 0 ~~\textrm{and} ~~U\cap \Gamma(v) \neq \varnothing \} $$
 be the sets in $z$'s support that contain neighbours of $v$.  Note
 that $U\in \U_v$ could contain Steiner vertices on which the lifting
 procedure has already taken place. However, by \eqref{eq:pi} and by
 \prettyref{lemma:3lps} the multi-family $\{U\cap R: U\in\U_v\}$ is
 laminar.  In the lifting process, we will transfer $x_U$ units of the
 $z_U$ units of dual of each set $U \in \U_v$ to the set $U'=U \cup
 \{v\}$; this decreases the dual load (LHS of \eqref{eq:LP-BD1}) on
 arcs from $U \cap \Gamma(v)$ to $v$ (e.g.~$uv$ in \prettyref{fig:lift}) and
 increases the dual load on arcs from $v$ to $\Gamma(v) \bs U$
 (e.g.~$vu'$ in the figure).  The following system of inequalities
 describes the set of feasible liftings.

 \begin{align}
   \forall U\in \U_v: &\qquad x_U \le z_U \tag{L1} \label{eq:L1}\\
   \forall u \in \Gamma(v): &\qquad \sum_{U: u\in U} (z_U - x_U) \le c_{uv}  \tag{L2} \label{eq:L2}\\
   \forall u \in \Gamma(v):&\qquad \sum_{U: u\notin U} x_U \le c_{uv}  \tag{L3} \label{eq:L3}
\end{align}

\begin{claim}\label{claim:primal}
  If \eqref{eq:L1}, \eqref{eq:L2}, \eqref{eq:L3} have a feasible
  solution $x \ge 0$, then the lifting procedure can be performed at Steiner
  vertex $v$, while maintaining the projection invariant property.
\end{claim}
\begin{proof}
  Define the new solution to be $z_U := z_U - x_U$, and, $z_{(U\cup
    v)} := x_U$, for all $U\in \U_v$, and $z_U$ remains unchanged for
  all other $U$. It is easy to check that all edges which were
  satisfied remain satisfied, and \eqref{eq:L2} and \eqref{eq:L3}
  imply that all edges incident to $v$ are satisfied. Also note that
  the projection invariant property is maintained.
\end{proof}

By Farkas' lemma, if \eqref{eq:L1}, \eqref{eq:L2}, \eqref{eq:L3} do {\em
  not} have a feasible solution $x \ge 0$, then there exist non-negative multipliers ---
$\lambda_U$ for all $U\in \U_v$, and $\alpha_u,\beta_u$ for all $u\in
\Gamma(v)$ --- satisfying the following dual set of linear inequalities:
\begin{align}
\sum_{U\in \U_v} \lambda_Uz_U + \sum_{u\in \Gamma(v)} \alpha_u \bigl(c_{uv}
- \sum_{U:u\in U} z_U\bigr) + \sum_{u\in \Gamma(v)} \beta_u c_{uv}
& \quad < \quad 0 \label{eq:D1} \tag{D1} \\
\forall U \in \U_v:  \lambda_U - \sum_{u\in U} \alpha_u +
\sum_{u\notin U} \beta_u & \quad \ge \quad 0 \label{eq:D2} \tag{D2}
\end{align}

As a technicality, note that the sub-system
$\{\eqref{eq:L1},\eqref{eq:L2},x\ge0\}$ is feasible --- take $x=z$.
Thus any $\alpha, \beta, \lambda$ satisfying \eqref{eq:D1} and
\eqref{eq:D2} has $\sum_u \beta_u > 0$, so by dividing all $\alpha,
\beta, \lambda$ by $\sum_i \beta_i$, we may assume without loss of
generality that
\begin{align}
  \sum_{u \in \Gamma(v)} \beta_u = 1 \tag{D3}. \label{eq:D3}
\end{align}
Subtracting \eqref{eq:D3} from \eqref{eq:D2} allows us to rewrite
the latter set of constraints conveniently as
\begin{align}
\forall U \in \U_v: &\qquad \lambda_U - \sum_{u\in U} (\alpha_u + \beta_u) + 1\ge 0  \tag{D2'}. \label{eq:D2'}
\end{align}

The following claim shows that \eqref{eq:L1}, \eqref{eq:L2},
\eqref{eq:L3} does have a feasible solution, and thus by
\prettyref{claim:primal}, lifting can be done, which completes the
proof of \prettyref{lem:main-qb}.

\begin{claim}\label{claim:dual}
There exists no feasible solution to $\{\alpha,\beta,\lambda \ge 0: \eqref{eq:D1},\eqref{eq:D2'}, \textrm{and } \eqref{eq:D3}\}$.
\end{claim}
\begin{proof}
  Consider the linear program which minimizes the LHS of \eqref{eq:D1}
  subject to the constraints \eqref{eq:D2'} and \eqref{eq:D3}. We show that the LP
  has value at least $0$, which will complete the proof.

  Let $(\lambda^*,\alpha^*,\beta^*)$ be an
  optimal solution to the LP. In \prettyref{lemma:tu} we will show that the
  constraint matrix of the LP is
  totally unimodular; hence, since the right-hand side of the given
  system is integral, we may assume that $\lambda^*, \alpha^*$, and
  $\beta^*$ are non-negative and integral. From \eqref{eq:D3} we infer
  \begin{equation}
\textrm{There is a unique $\bar{u} \in \Gamma(v)$ for which
      $\beta^*_{\bar{u}}=1$; for all $u \neq \bar{u}$, $\beta^*_u=0$.}\label{eq:olda}\end{equation}
 Moreover, since each $\lambda_U$ appears only in the two constraints \eqref{eq:D2'} and $\lambda_U \ge 0$, and since $\lambda_U$ has nonnegative coefficient in the objective, we may assume \begin{equation} \lambda^*_U = \lambda^*_U(\alpha^*, \beta^*) := \max\{\sum_{u\in U} (\alpha^*_u + \beta^*_u) - 1, 0\} \label{redu}\end{equation}
  for all $U$.

 Next, we establish the following:
% \begin{enumerate}
%    \item[{[a]}] There is at most one $\bar{u} \in \Gamma(v)$ for which
%      $\beta^*_{\bar{u}}=1$.
  \begin{equation}
%    \item[{[b]}]
    \textrm{$\alpha^*_u+\beta^*_u \in \{0,1\}$ for all $u \in \Gamma(v)$.}\label{eq:oldb}\end{equation}
%  \end{enumerate}
%  To see $[\textrm{a}]$, note that, given any feasible solution
%  $\lambda^*,\alpha^*,\beta^*$ with $\sum_u \beta^*_u > 1$, decreasing $\beta^*_u$ and increasing
%  $\alpha^*_u$ by the same amount for any $u$ maintains the feasibility
%  of \eqref{eq:D2'}, and does not
%  increase the objective value of the solution (since $z \ge 0$). Applying
%  this operation repeatedly, property $[\textrm{a}]$ holds.
  Suppose for the sake of contradiction that property \eqref{eq:oldb} does not hold for our solution.
  Let $u$ be such that $\alpha^*_u+\beta^*_u \geq 2$. By \eqref{eq:olda}, $\alpha^*_u \ge 1$. We propose the following
  update to our solution:
  decrease $\alpha^*_u$ by $1$ (which by \eqref{redu} will decrease $\lambda^*_U$ by
  $1$ for all $U\in \U_v$). This maintains the
  feasibility of \eqref{eq:D2'}, and the objective value
  decreases by
  $$ \sum_{U\in \U_v : u \in U}z_U + (c_{uv} - \sum_{u\in U} z_U) $$
  which is non-negative as $c \geq 0$.
  By repeating this operation, we may clearly ensure property \eqref{eq:oldb}.

  Let $K\subseteq \Gamma(v)$ be the set $\{u \mid \alpha^*_u + \beta^*_u =
  1\}$ and recall $\bar{u}$ is the unique terminal with
  $\beta^*_{\bar{u}}=1$; $\bar{u}$ is clearly a member of $K$.
  At $(\alpha^*, \beta^*, \lambda^*)$, we evaluate the objective and collect like terms to get value
  \begin{align*}
  \sum_{U\in \U_v} z_U\rho(U\cap K) + \sum_{u\in K\setminus \bar{u}}
  (c_{uv} - \sum_{U: u\in U}z_U) + c_{\bar{u}v}
  &=
  \sum_{u\in K} c_{uv} + \sum_{U\in \U_v}z_U(\rho(U\cap K) -|(K \bs \bar{u}) \cap U|)  \\
  &= \sum_{u\in K} c_{uv} - \sum_{U \in \U_v: U \cap K \neq \varnothing, \bar{u} \not\in U}z_U
  \end{align*}
  where the last equality follows by considering cases. Finally, combining the fact that
  $\sum_{u\in K} c_{uv} \ge C_K$ (since these edges form one possible full component on terminal set $K$) together with \eqref{eq:LP-A2} for the
 pair $(K,\bar{u})$, it follows that the LP's optimal value is non-negative as needed.
\end{proof}

\begin{lemma}\label{lemma:tu}
The incidence matrix defined by \eqref{eq:D2'} and \eqref{eq:D3} is totally unimodular.
\end{lemma}
\begin{proof}
  The incidence matrix has $|\U_v|+1$ rows ($|\U_v|$ corresponding to
  \eqref{eq:D2'} and one last row corresponding to \eqref{eq:D3}) and
  $|\U_v| + 2|\Gamma(v)|$ columns. Furthermore, the columns corresponding to $\alpha_u$'s are same
  as those corresponding to $\beta_u$'s, except for the last row,
  where there are $0$'s in the $\alpha$-columns and $1$'s in the
  $\beta$-columns.

  To show that this matrix is totally unimodular
  we use Ghouila-Houri's characterization of total unimodularity (e.g.\ see \cite[Thm. 19.3]{Sc86}):
  \begin{theorem}[Ghouila-Houri 1962] A matrix is totally unimodular iff the following holds for \emph{every} subset $\mathcal{R}$ of rows: we can assign weights $w_r \in \{-1, +1\}$ to each row $r \in \mathcal{R}$ such that $\sum_{r\in \mathcal{R}} w_r r$ is a $\{0, \pm1\}$-vector.\end{theorem}

  Note that we can safely ignore the columns corresponding to
  variables $\lambda_U$ for sets $U \in \U_v$, since each of them contains a
  single $1$ occurring in constraint \eqref{eq:D2'} for set
  $U$.

  The row subset $\mathcal{R}$ corresponds to a subset of $\U_v$ --- which we will denote $\mathcal{R} \cap \U_v$ --- plus possibly the single row corresponding to \eqref{eq:D3}. Each row in $\mathcal{R} \cap \U_v$ has its values determined by the characteristic vector of $U\cap \Gamma(v)$. So long as any set appears more than once in $\{U\cap \Gamma(v) \mid U \in \mathcal{R} \cap \U_v\}$ we can assign one copy weight $+1$ and the other copy weight $-1$; these rows cancel out. Thus, henceforth we assume $\{U\cap \Gamma(v) \mid U \in \mathcal{R} \cap \U_v\}$ has no duplicate sets.

  There is a standard
  representation of a laminar family as a forest of rooted trees,
  where there is a node corresponding to each set, with containment in the family corresponding to ancestry in the forest.  Given the
  forest for the laminar family $\{U\cap \Gamma(v) \mid U \in \mathcal{R} \cap \U_v\}$, the assignment of weights to the rows of the matrix is as
  follows.  Let the root nodes of all trees be at height $0$ with
  height increasing as one goes to children nodes.  Give weight $-1$
  to rows corresponding to nodes at even height, and weight $+1$ to rows
  corresponding to nodes at odd height. If $\mathcal{R}$ contains
  the row corresponding to \eqref{eq:D3}, give it weight $+1$.

  Finally, let us argue that these weights have the needed property.
  Consider first a column corresponding to $\alpha_u$ for any $u$.
  The rows of $\mathcal{R}$ with $1$ in this column form a path, from the largest set containing $u$ (which is a root node) to the smallest set containing $u$.
  The weighted sum in this column is an alternating sum $-1+1-1+1\dotsb$, which is either $-1$ or $0$, which is in $\{0, \pm1\}$ as needed. Second, in a column for some $\beta_u$, if $\mathcal{R}$ doesn't contain (resp.\ contains) the row corresponding to \eqref{eq:D3}, the weighted sum is the same as for $\alpha_u$ (resp.\ plus 1); in either case its weighted sum is in $\{0, \pm1\}$ as needed.
\end{proof}

This finishes the proof of \prettyref{lem:main-qb}, and hence also
that of \prettyref{theorem:lifting}.
\end{proof}

}{
  For lack of space, we present only sketches for our main equivalence
  results in this extended abstract, and refer the reader to
  \cite{CKP09} for details.

  \begin{theorem}\label{theorem:p-pu}
    The LPs \eqref{eq:LP-P2} and \eqref{eq:LP-PU} have the same optimal
    value.
  \end{theorem}
  \noindent{\em Proof sketch.} To show this, it suffices
    to find an optimum solution of \eqref{eq:LP-PU} which
    satisfies the equality in \eqref{eq:LP-P2}; i.e., we want
    to find a solution for which the maximal-rank partition
    $\overline{\pi}$ is tight.  We pick the optimum
    solution to \eqref{eq:LP-PU} which minimizes the sum $\sum_{K\in
      \K} x_K|K|$.  Using \prettyref{ppty:quote}, we show
    that either $\overline{\pi}$ is tight or there is a {\em
      shrinking} operation which decreases $\sum_{K\in \K} x_K|K|$
    without increasing the cost. Since the latter is impossible, the
    theorem is proved.

  \begin{theorem}\label{theorem:spe}
    The feasible regions of \eqref{eq:LP-P2} and \eqref{eq:LP-S} are the
    same.
  \end{theorem}
  \noindent{\em Proof sketch.} We show that the inequalities defining
    \eqref{eq:LP-P2} are valid for \eqref{eq:LP-S}, and
    vice-versa. Note that both have the same equality and
    non-negativity constraints. To show that the partition inequality
    of \eqref{eq:LP-P2} for $\pi$ holds for any $x \in
    \eqref{eq:LP-S},$ we use the subtour inequalities in
    \eqref{eq:LP-S} for every part of $\pi$. For the other direction,
    given any subset $S\subseteq R$, we invoke the inequality in
    \eqref{eq:LP-P2} for the partition $\pi := \{\{S\} \textrm{ as one part and
      the remaining terminals as singletons}\}$.

  \begin{theorem}\label{theorem:lifting}
    On quasibipartite Steiner tree instances,
    $\OPT\eqref{eq:LP-B} \ge \OPT\eqref{eq:LP-PUDir}$.
\end{theorem}
\noindent{\em Proof sketch.} We look at
  the duals of the two LPs and we show $\OPT\eqref{eq:LP-BD} \ge
  \OPT\eqref{eq:LP-A}$ in quasibipartite instances.  Recall that the
  support of a solution to \eqref{eq:LP-A} is the family of sets with
  positive $z_U$. A family of sets is called \emph{laminar} if for any
  two of its sets $A,B$ we have $A\subseteq B, B\subseteq A$, or
  $A\cap B=\varnothing$. The following fact follows along the standard line of ``set uncrossing" argumentation.

  \begin{lemma} \label{lemma:3lps}
    There is an optimal solution to
    \eqref{eq:LP-A} with laminar support.
  \end{lemma}

  Given the above result, we may now assume that we have a solution
  $z$ to \eqref{eq:LP-A} whose support is laminar. The heart of the
  proof of \prettyref{theorem:lifting} is to show that $z$ can be
  converted into a feasible solution to \eqref{eq:LP-BD} of the same
  value.

  Comparing \eqref{eq:LP-A} and \eqref{eq:LP-BD} one first notes that
  the former has a variable for every valid subset of the terminals,
  while the latter assigns values to all valid subsets of the entire
  vertex set.  We say that an edge $uv$ is \emph{satisfied} for a
  candidate solution $z$, if both a) $\sum_{U:u\in U, v\notin U} z_{U}
  \le c_{uv}$ and b) $\sum_{U:v\in U, u\notin U} z_{U} \le c_{uv}$
  hold; $z$ is then feasible for \eqref{eq:LP-BD} if {\em all} edges
  are satisfied.

  Let $z$ be a feasible solution to \eqref{eq:LP-A}.
  One easily verifies that all terminal-terminal edges are
  satisfied. On the other hand, terminal-Steiner edges may
  initially not be satisfied; e.g., consider the Steiner vertex
  $v$ and its neighbours depicted in \prettyref{fig:lift} below.
  Initially, none of the sets in $z$'s support contains $v$, and
  the load on the edges incident to $v$ is quite {\em skewed}:
  the left-hand side of condition a) above may be large, while
  the left-hand side of condition b) is initially $0$.

 \piccaptioninside
  \piccaption{\label{fig:lift} Lifting variable $z_U$.}
  \parpic(4.5cm,4cm)[fr]{\includegraphics[scale=.7]{lift.eps}}

  To construct a valid solution for \eqref{eq:LP-BD}, we therefore
  {\em lift} the initial value $z_S$ of each terminal subset $S$ to
  supersets of $S$, by adding Steiner vertices. The lifting
  procedure processes each Steiner vertex $v$ one at a time; when
  processing $v$, we change $z$ by moving dual from some sets $U$ to
  $U \cup \{v\}$. Such a dual transfer decreases the left-hand side
  of condition a) for edge $uv$, and increases the
  (initially $0$) left-hand sides of condition b) for edges connecting $v$ to
  neighbours other than $v$.

  We are able to show that there is a way of carefully lifting duals
  around $v$ that ensures that all edges incident to $v$ become
  satisfied. The definition of our procedure will ensure
  that these edges remain satisfied for the rest of the lifting
  procedure. Since there are no Steiner-Steiner edges, all edges will
  be satisfied once all Steiner vertices are processed.

  Throughout the lifting procedure, we will maintain that $z$ remains
  unchanged, when projected to the terminals. The main consequence
  of this is that the objective value
  $\sum_{U\subseteq V} z_U$ remains constant throughout, and
  the objective value of $z$ in \eqref{eq:LP-BD} is not affected
  by the lifting. This yields \prettyref{theorem:lifting}.

}

\section{Improved Integrality Gap Upper Bounds}\label{sec:gapbounds}

\myifthen{ We first show the improved bound of $73/60$ for uniformly
  quasibipartite graphs.  We then show the $(2\sqrt{2} - 1) \doteq
  1.828$ upper bound on general graphs, which contains the main ideas,
  and then end by giving a $\sqrt{3} \doteq 1.729$ upper bound.
}{
  In this extended abstract, we show the improved bound of $73/60$ for
  uniformly quasibipartite graphs, and due to space restrictions, we
  only show the weaker $(2\sqrt{2} - 1) \doteq 1.828$ upper bound on
  general graphs.
}

\subsection{Uniformly Quasibipartite Instances}

Uniformly quasibipartite instances of the Steiner tree problem are
quasibipartite graphs where the cost of edges incident on a Steiner
vertex are the same. They were first studied by Gr\"opl et
al.~\cite{GH+02}, who gave a $73/60$ factor approximation
algorithm.
\myifthen{In the following, we show that the cost of the returned
tree is no more than than $\7 \OPT\eqref{eq:LP-PU}$, which
upper-bounds the integrality gap by $\7$.

}{}
We start by describing the algorithm of Gr\"opl et al.~\cite{GH+02} in
terms of full components.  A collection $\K'$ of full components is
acyclic if there is no list of $t > 1$ distinct terminals and
hyperedges in $\K'$ of the form $r_1 \in K_1 \ni r_2 \in K_2 \dotsb
\ni r_t \in K_t \ni r_1$ --- i.e.~there are no \emph{hypercycles}.

\vspace{3ex}\noindent
\begin{boxedminipage}{\algobox}
Procedure \mss
\begin{algorithmic}[1]
  \STATE Initialize the set of acyclic components $\L$ to $\varnothing$. \\
  \STATE Let $L^*$ be a minimizer of $\frac{C_L}{|L| - 1}$ over all full components $L$ such that $|L| \ge 2$ and $L\cup\L$ is acyclic. \\
  \STATE Add $L^*$ to $\L$. \\
  \STATE Continue until $(R, \L)$ is a hyper-spanning tree and return $\L$.
\end{algorithmic}
\end{boxedminipage}
\vspace{0.75ex}

\begin{theorem} \label{theorem:uniformly}
On a uniformly quasibipartite instance \mss\ returns a
Steiner tree of cost at most $\7\OPT\eqref{eq:LP-PU}$.
\end{theorem}
\myifthen{
\begin{proof}
  Let $t$ denote the number of iterations and $\L :=
  \{L_1,\ldots,L_t\}$ be the ordered sequence of full components
  obtained.  We now define a dual solution to \eqref{eq:LP-PUD}. Let
  $\pi(i)$ denote the partition induced by the connected components of
  $\{L_1, \dotsc, L_i\}$. Let $\theta(i)$ denote $C_{L_i}/(|L_i| - 1)$
  and note that $\theta$ is nondecreasing. Define $\theta(0)=0$ for
  convenience. We define a dual solution $y$ with
  $$y_{\pi(i)} = \theta(i+1)-\theta(i)$$
  for $0 \le i < t$, and all other coordinates of $y$ set to zero; $y$
  is not generally feasible, but we will scale it down to make it
  so. By evaluating a telescoping sum, it is not hard to find that
  $\sum_i y_{\pi(i)} (r(\pi(i))-1) = C(\L)$.  In the rest of the proof
  we will show for any $K\in \K$, $\sum_i y_{\pi(i)} \rc^{\pi(i)}_K
  \le 73/60\cdot C_K$ --- by scaling, this also proves that
  $\frac{60}{73} y$ is a feasible dual solution, and hence completes
  the proof.

  Fix any $K \in \K$ and let $|K| = k$. Since the instance in question
  is uniformly quasi-bipartite, the full component $K$ is a star with
  a Steiner centre and edges of a fixed cost $c$ to each terminal in
  $K$. For $1 \le i < k$, let $\tau(i)$ denote the last iteration $j$
  in which $\rc_K^{\pi(j)} \ge k-i$. Let $K_i$ denote any subset of
  $K$ of size $k-i+1$ such that $K_i$ contains at most one element
  from each part of $\pi(\tau(i))$; i.e., $|K_i| = k-i+1$ and
  $\rc_{K_i}^{\pi(\tau(i))} = k-i$.

  Our analysis hinges on the fact that $K_i$ was a valid choice for
  $L_{\tau(i)+1}$. More specifically, note that $\{L_1, \dotsc,
  L_{\tau(i)}, K_i\}$ is acyclic, hence by the greedy nature of the
  algorithm, for any $1 \le i < k,$ $$\theta(\tau(i)+1) =
  C_{L_{\tau(i)+1}}/(|L_{\tau(i)+1}|-1) \le C_{K_i}/(|K_i|-1) \le
  \frac{c \cdot (k-i+1)}{k-i}.$$ Moreover, using the definition of
  $\tau$ and telescoping we compute
  $$\sum_\pi y_\pi \rc_K^\pi = \sum_{i=0}^{t-1} (\theta(i+1)-\theta(i))\rc_K^{\pi(i)} = \sum_{i=1}^{k-1} \theta(\tau(i)+1)
  \le \sum_{i=1}^{k-1} \frac{c \cdot (k-i+1)}{k-i} = c\cdot
  (k-1+H(k-1)),$$ where $H(\cdot)$ denotes the harmonic
  series. Finally, note that $(k-1+H(k-1)) \le \7k$ for all $k \ge 2$
  (achieved at $k=5$). Therefore, $\frac{60}{73}y$ is a valid solution
  to \eqref{eq:LP-PUD}.
  \end{proof}
  }
{
 \noindent{\emph{Proof sketch.}} Let $t$ denote the number of iterations and $\L :=
  \{L_1,\ldots,L_t\}$ be the ordered sequence of full components
  obtained.  We now define a dual solution $y$ to
  \eqref{eq:LP-PUD}. Let $\pi(i)$ denote the partition induced by the
  connected components of $\{L_1, \dotsc, L_i\}$. Let $\theta(i)$
  denote $C_{L_i}/(|L_i| - 1)$ and note that $\theta$ is
  nondecreasing. Define $\theta(0)=0$ for convenience. We define a
  dual solution $y$ with
  $$y_{\pi(i)} = \theta(i+1)-\theta(i)$$
  for $0 \le i < t$, and all other coordinates of $y$ set to zero. It
  is straightforward to verify that the objective value $\sum_i y_{\pi(i)} (r(\pi(i))-1)$ of $y$ in \eqref{eq:LP-PUD} equals $C(\L)$.

  The key is to show that for all $K \in \K$,
  \begin{equation}
  \label{eq:bamz}
  \sum_i y_{\pi(i)} \rc^{\pi(i)}_K \le
  (|K|-1+H(|K|-1))/|K|\cdot C_K,
  \end{equation}
  where $H$ denotes the harmonic series; this is obtained by using the greedy nature of the algorithm and the
  fact that, in uniformly quasi-bipartite graphs, $C_{K'} \le C_K
  \frac{|K'|}{|K|}$ whenever $K' \subset K$. Now,
  $(|K|-1+H(|K|-1))/|K|$ is always at most $\frac{73}{60}$. Thus \eqref{eq:bamz} implies that $\frac{60}{73}\cdot y$
  is a feasible dual solution, which completes the proof.
}

\comment{This method also gives a 5/4 integrality gap bound on
  instances where every full component has size at most 3, and an
  $H(t-1)$ integrality gap bound for \eqref{eq:LP-PU} on general
  hypergraphs of maximum hyperedge size $t$ (i.e.\ ones not obtained
  from instances of the Steiner tree problem) --- see \cite{CKP09}.}

\comment{we also get an integrality gap bound of $H(t-1)$ on
  \eqref{eq:LP-PU} as an LP relaxation for the min-cost spanning
  sub-hypergraph problem when all hyperedges have size at most
  $t$. This complements the observation by Baudis et al.~\cite{BG+00}
  that \mss\ has approximation ratio $H(t-1)$, which is in turn a
  generalization of the submodular set cover framework of
  Wolsey~\cite{Wo82}. This is nearly best possible for $t \ge 4$ since
  ``set cover with maximum set size $k$" to can be reduced to
  ``spanning connected hypergraph with maximum edge size $k+1$" by
  creating a new root vertex and adding it to all sets. This set cover
  problem is \APX-hard for $k \ge 3$ and Trevisan~\cite{Trev01} showed
  $\ln k - O(\ln \ln k)$ inapproximability unless \PP=\NP.

  We also can extend these results, employing computational power, to
  get an integrality gap bound $\beta_r$ in the case of Steiner tree
  instances with at most $r$ terminals per full component. We obtain
  integrality gap upper bounds $\beta_3 = 5/4, \beta_4 = 11/8, \beta_5
  = 119/82, \beta_6 = 3/2$ respectively.}

\subsection{General graphs}
\def\drop{{\tt drop}} \def\gain{{\tt gain}}

\myifthen{
We start with a few definitions and notations in order to prove the
$2\sqrt{2}-1$ and $\sqrt{3}$ integrality gap bounds on \eqref{eq:LP-PU}. Both results use
similar algorithms, and the latter is a more complex version of the
former.
}{}
For conciseness we let a ``graph" be a triple $G = (V, E, R)$ where $R \subset V$ are $G$'s terminals. In the following, we let $\mtst(G; c)$ denote the minimum
\emph{terminal spanning tree}, i.e.~the minimum spanning tree of the terminal-induced subgraph $G[R]$ under edge-costs $c : E \to \R$. We will abuse notation and let $\mtst(G; c)$ mean both the tree and its cost under $c$.

When contracting an edge $uv$ in a graph, the new merged node resulting from contraction is defined to be a terminal iff at least one of $u$ or $v$ was a terminal; this is natural since a Steiner tree in the new graph is a minimal set of edges which, together with $uv$, connects all terminals in the old graph. Our \myifthen{algorithm performs}{algorithm performs} contraction, which may introduce parallel edges, but one may delete all but the cheapest edge from each parallel class without affecting the analysis.

Our \myifthen{first}{} algorithm proceeds in stages. In each stage we apply the operation $G \mapsto G/K$ which denotes contracting all edges in some full component $K$. To describe and analyze the algorithm we introduce some notation. For a minimum terminal spanning tree $T=\mtst(G;c)$
define $\drop_{T}(K;c) := c(T) - \mtst(G/K;c)$. We also define
$\gain_{T}(K;c):= \drop_{T}(K) - c(K)$, where $c(K)$ is the cost of
full component $K$.  A tree $T$ is called \emph{gainless} if for every
full component $K$ we have $\gain_T(K;c) \le 0$.  The following useful
fact is implicit in~\cite{KPT09} (see also
\myifthen{\prettyref{app:app2}}{\cite{CKP09}}).

\setcounter{thm_locopt}{\value{theorem}}
\begin{theorem}[Implicit in \cite{KPT09}] \label{thm:locopt}
  If $\mtst(G; c)$ is gainless, then
  $\OPT\eqref{eq:LP-PU}$ equals the cost of $\mtst(G; c)$.
\end{theorem}

 We now give the \myifthen{first}{} algorithm and its analysis, which
  uses a reduced cost trick introduced by Chakrabarty et
  al.\cite{CDV08}.

\medskip

\noindent
\begin{boxedminipage}{\algobox}
Procedure {\sc Reduced One-Pass Heuristic}
\begin{algorithmic}[1]
\STATE Define costs $c'_e$ by $c'_e := c_e/\sqrt{2}$ for all
  terminal-terminal edges $e$, and $c'_e = c_e$ for all other edges.  Let $G_1 := G,$ $T_i := \mtst(G_i; c')$, and $i:=1$.
\STATE The algorithm
  considers the full components in any order. When we examine a full component
  $K$, if $\gain_{T_i}(K;c') > 0$, let
  $K_i := K$, $G_{i+1} := G_i/K_i$, $T_{i+1} :=
  \mtst(G_{i+1};c')$, and $i:=i+1$.
\STATE Let $f$ be the final value of $i$. Return the tree $T_{alg} :=
  T_f \cup \bigcup_{i=1}^{f-1} K_i$.
\end{algorithmic}
\end{boxedminipage}
\medskip

\noindent
Note that the full components are scanned in {\em any} order and they
are not examined a priori. Hence the algorithm works just as well if
the full components arrive ``online," which might be useful for some
applications.

\begin{theorem}\label{theorem:bound1828}
$c(T_{alg}) \leq (2\sqrt{2} - 1) \OPT\eqref{eq:LP-PU}$.
\end{theorem}
\begin{proof}
  First we claim that $\gain_{T_f}(K;c') \le 0$ for all $K$. To see this there are two cases. If $K=K_i$ for some $i$, then we immediately see that $\drop_{T_j}(K) = 0$ for all $j > i$ so $\gain_{T_f}(K) = -c(K) \le 0$. Otherwise (if for all $i,$ $K \neq K_i$) $K$ had nonpositive gain when examined by the algorithm; and the well-known \emph{contraction lemma} (e.g., see~\cite[\S 1.5]{GH+01b}) immediately implies that $\gain_{T_i}(K)$ is nonincreasing in $i$, so $\gain_{T_f}(K) \le 0$.

  By \prettyref{thm:locopt},
  $c'(T_f)$ equals the value of \eqref{eq:LP-PU} on the
  graph $G_f$ with costs $c'$.  Since $c' \le c$, and since at each
  step we only contract terminals, the value of this optimum must be
  at most $\OPT\eqref{eq:LP-PU}$. Using the fact that $c(T_f) =
  \sqrt{2}c'(T_f)$, we get
\begin{align}\label{eq:dracula}
c(T_f) = \sqrt{2}c'(T_f) \le \sqrt{2} \OPT\eqref{eq:LP-PU}
\end{align}
\noindent
Furthermore, for every $i$ we have $\gain_{T_i}(K_i;c') > 0$, that is,
$\drop_{T_i}(K_i;c') > c'(K) = c(K)$. The equality follows since $K$
contains no terminal-terminal edges. However, $\drop_{T_i}(K_i;c') =
\frac{1}{\sqrt{2}} \drop_{T_i}(K_i;c)$ because all edges of $T_i$ are
terminal-terminal.  Thus, we get for every $i=1$ to $f$,
~$\drop_{T_i}(K_i;c) > \sqrt{2}\cdot c(K_i)$.

Since $\drop_{T_i}(K_i;c) := \mtst(G_i;c) - \mtst(G_{i+1};c)$, we have
$$\sum_{i=1}^{f-1} \drop_{T_i}(K_i;c)=\mtst(G;c) - c(T_f).$$
Thus, we have
\myifthen{
\begin{equation*}
  \sum_{i=1}^{f-1} c(K_i) \le \frac{1}{\sqrt{2}} \sum_{i=1}^f
  \drop_{T_i}(K_i;c) = \frac{1}{\sqrt{2}} (\mtst(G;c) - c(T_f))
  \le \frac{1}{\sqrt{2}}(2\OPT\eqref{eq:LP-PU} - c(T_f))
\end{equation*}
}
{
\begin{align*}
  \sum_{i=1}^{f-1} c(K_i) \le \frac{1}{\sqrt{2}} \sum_{i=1}^f
  \drop_{T_i}(K_i;c) &= \frac{1}{\sqrt{2}} (\mtst(G;c) - c(T_f)) \\
  &\le \frac{1}{\sqrt{2}}(2\OPT\eqref{eq:LP-PU} - c(T_f))
\end{align*}
}
where we use the fact that $\mtst(G, c)$ is at most twice
$\OPT\eqref{eq:LP-PU}$\footnote{This follows using standard arguments, and can be
  seen, for instance, by applying \prettyref{thm:locopt} to the
  cost-function with all terminal-terminal costs divided by 2, and
  using short-cutting.}. Therefore
$$c(T_{alg}) = c(T_f) + \sum_{i=1}^{f-1} c(K_i) \le \Bigl(1 - \frac{1}{\sqrt{2}}\Bigr) c(T_f) + \sqrt{2}\OPT\eqref{eq:LP-PU}.$$
Finally, using $c(T_f) \le \sqrt{2}\OPT\eqref{eq:LP-PU}$ from
\eqref{eq:dracula}, the proof of \prettyref{theorem:bound1828} is
complete. \end{proof}

\myifthen{
\subsubsection{Improving to $\sqrt{3}$}
\def\loss{{\tt loss}} To get the improved factor of $\sqrt{3}$, we use
a more refined iterated contraction approach. The crucial new concept
is that of the {\em loss} of a full component, introduced by Karpinski
and Zelikovsky \cite{KZ97}. The intuition is as follows. In each
iteration, the $(2\sqrt{2}-1)$-factor algorithm contracts a full
component $K$, and thus commits to include $K$ in the final solution;
the new algorithm makes a smaller commitment, by contracting a
\emph{subset} of $K$'s edges, which allows for a possibility of better
recovery later.

Given a full component $K$ (viewed as a tree with leaf set $K$ and
internal Steiner nodes), $\loss(K)$ is defined to be the minimum-cost
subset of $E(K)$ such that $(V(K), \loss(K))$ has at least one
terminal per connected component --- i.e.~the cheapest way in $K$ to
connect each Steiner node to the terminal set. We also use $\loss(K)$
to denote the total \emph{cost} of these edges. Note that no two
terminals are connected by $\loss(K)$.
%The {\em loss contracted} tree  $T/\loss(K)$ is obtained by contracting the edges in $\loss(K)$. Note that $T(K)$ is
%a tree spanning the terminals in $K$ and $c(T(K)) = c(K) - \loss(K)$.
A very useful theorem of Karpinski and Zelikovsky \cite{KZ97}
is that for any full component $K$, $\loss(K) \le c(K)/2$.

Now we have the ingredients to give our new algorithm. In the
description below, $\alpha > 1$ is a parameter (which will be set to
$\sqrt{3}$). In each iteration, the algorithm contracts the loss of a
single full component $K$ (we note it follows that the terminal set has constant size over all iterations).
\medskip
%Notice that
%the contraction of loss edges may decrease the distance between
%terminals; in the following algorithm we therefore keep track of
%the current contracted graph and its associated cost function.

\noindent
\begin{boxedminipage}{\algobox}
Procedure {\sc Reduced One-Pass Loss-Contracting Heuristic}
\begin{algorithmic}[1]
\STATE Initially $G_1 := G$, $T_1 := \mtst(G;c)$, and $i:=1$.
\STATE
  The algorithm considers the full components in any order.  When we examine a full component
  $K$, if
$$\gain_{T_i}(K;c) > (\alpha - 1)\loss(K),$$
let
  $K_i := K$, $G_{i+1} := G_i/\loss(K_i)$, $T_{i+1} :=
  \mtst(G_{i+1};c)$, and $i:=i+1$.
\STATE Let $f$ be the final value of $i$. Return the tree $T_{alg} :=
 T_f \cup \bigcup_{i=1}^{f-1} \loss(K_i).$
\end{algorithmic}
\end{boxedminipage}
\medskip

We now analyze the algorithm.
\noindent
\begin{claim}\label{claim:tf}
$c(T_f) \le (\frac{1+\alpha}{2}) \OPT\eqref{eq:LP-PU}$.
\end{claim}
\begin{proof}
Using the contraction lemma again, $\gain_{T_f}(K;c) \le (\alpha - 1)\loss(K)$ for all $K$, so
\begin{align}\label{eq:wtf}
\drop_{T_f}(K;c) \le c(K) + (\alpha - 1) \loss(K) = c(K) + (\alpha - 1)\loss(K) \le \Big(\frac{1+\alpha}{2}\Big)c(K)
\end{align}
since $\loss(K) \le c(K)/2$.

To finish the proof of \prettyref{claim:tf}, we proceed as in the
proof of Equation \eqref{eq:dracula}. Define $c'_e :=
c_e/(\frac{1+\alpha}{2})$ for all edges $e$ which join two vertices of the original terminal set $R$, and $c'_e = c_e$ for all other edges. Note that \eqref{eq:wtf} implies that $T_f$
is gainless with respect to $c'$. Thus, by \prettyref{thm:locopt},
the value of LP \eqref{eq:LP-PU} on $(G_f, c')$ equals
$c'(T_f)$. Since we only reduce costs (as $\alpha \ge 1$), this
optimum is no more than the original $\OPT\eqref{eq:LP-PU}$ giving us
$c'(T_f) \le \OPT\eqref{eq:LP-PU}$. Now using the definition of
$c'$, the proof of the claim is complete.
\end{proof}

\begin{claim}\label{claim:drop}
For any $i\ge 1$, we have $c(T_i) - c(T_{i+1}) \ge \gain_{T_i}(K_i;c) + \loss(K_i)$.
\end{claim}
\begin{proof}
  Recall that $T_{i+1}$ is a minimum terminal spanning tree of
  $G_{i+1}$ under $c$. Consider the following other terminal
  spanning tree $T$ of $G_{i+1}$: take $T$ to be the union of $K_i /
  \loss(K_i)$ with $\mtst(G_i/K_i; c)$.  Hence $c(T_{i+1}) \le
  c(T) = \mtst(G_i/K_i; c) + c(K_i) - \loss(K_i)$. Rearranging,
  and using the definition of gain, we obtain:
  \begin{equation*}
    c(T_i) - c(T_{i+1}) \ge c(T_i) - \mtst(G_i/K_i; c) - c(K_i) + \loss(K_i) = \gain_{T_i}(K_i; c) +\loss(K_i),
  \end{equation*}
  and this completes the proof.
\end{proof}
\noindent
Now we are ready to prove the integrality gap upper bound of $\sqrt{3}$.
\begin{theorem}
$c(T_{alg}) \le \sqrt{3}\OPT\eqref{eq:LP-PU}$.
\end{theorem}
\begin{proof}
  By the algorithm, we have for all $i$ that $\gain_{T_i}(K_i) \ge
  (\alpha - 1)\loss(K_i)$, and thus $\gain_{T_i}(K_i; c) +\loss(K_i)
  \ge \alpha \loss(K_i)$. Thus, from \prettyref{claim:drop}, we get
$$\sum_{i=1}^{f-1} \loss(K_i) \le \frac{1}{\alpha}  \sum_{i=1}^{f-1} \Big(c(T_i) - c(T_{i+1}) \Big)$$
\noindent
The right-hand sum telescopes to give us $c(T_1) - c(T_f) = \mtst(G;c) - c(T_f)$. Thus,
\begin{align*}
c(T_{alg}) &= c(T_f) + \sum_{i=1}^{f-1} \loss(K_i) \le c(T_f) + \frac{1}{\alpha}(\mtst(G;c) - c(T_f)) = \frac{1}{\alpha}\mtst(G;c) + \frac{\alpha-1}{\alpha} c(T_f)  \\
&\le \Big(\frac{2}{\alpha} + \frac{(\alpha -1)(1+\alpha)}{2\alpha}\Big)\OPT\eqref{eq:LP-PU}
 = \Big(\frac{\alpha^2 + 3}{2\alpha}\Big)\OPT\eqref{eq:LP-PU}
\end{align*}
which follows from $\mtst(G;c)  \le 2\OPT\eqref{eq:LP-PU}$ and \prettyref{claim:tf}.
Setting $\alpha = \sqrt{3}$, the proof of the theorem is complete.
\end{proof}
}{}

\myifthen{
\section{Conclusion}
In this paper we looked at several hypergraphic LP relaxations for the
Steiner tree problem, and showed they all have the same objective
value. Furthermore, we noted some connections to the bidirected cut
relaxation for Steiner trees: although hypergraphic relaxations are
stronger than the bidirected cut relaxation in general, in
quasibipartite graphs all these relaxations are equivalent.  We
obtained structural results about the hypergraphic relaxations showing
that basic feasible solutions have sparse support. We also showed
improved upper bounds on the integrality gaps on the hypergraphic
relaxations via simple algorithms.

Reiterating the comments in Section \ref{sec:discussion}, the
hypergraphic LPs are powerful (e.g.~as evidenced by Byrka et
al.~\cite{BGRS10}) but may not be manageable for computational
implementation. Some interesting areas for future work include:
non-ellipsoid-based algorithms to solve the hypergraphic LPs in the
$r$-restricted setting; resolving the complexity of optimizing them in
the unrestricted setting; and directly using the bidirected cut
relaxation to achieve good results (e.g.~in quasi-bipartite
instances).
}{}
\bibliography{ckp}
\bibliographystyle{abbrv}
\myifthen{
\appendix
\section{Directed Hypergraph LP Relaxation}\label{app:reproof}
\begin{theorem}\label{theorem:franky}
For any Steiner tree instance, $\OPT\eqref{eq:LP-PU} = \OPT\eqref{eq:LP-PUDir}$.
\end{theorem}
\begin{proof}
First, we show $\OPT\eqref{eq:LP-PU} \le \OPT\eqref{eq:LP-PUDir}$. Consider a feasible solution $x$ to $\eqref{eq:LP-PUDir}$, and define a solution $x'$ to $\eqref{eq:LP-PU}$ by $x'_K = \sum_{i \in K} x_{K^i}$; informally, $x'$ is obtained from $x$ by ignoring the orientation of the hyperedges. Clearly $x'$ and $x$ have the same objective value. Further, $x'$ is feasible for $\eqref{eq:LP-PU}$; to see this, for any partition $\pi$, note that       \eqref{eq:LP-PU2} is implied by the sum of constraints \eqref{eq:LP-PUDir2} over $U$ set to those parts of $\pi$ not containing the root --- any orientation of a full component with rank contribution $t$ must leave at least $t$ parts.

To obtain the reverse direction $\OPT\eqref{eq:LP-PUDir} \le \OPT\eqref{eq:LP-PU}$, we use a similar strategy. We require some notation and a hypergraph orientation theorem of Frank et al.~\cite{FKK03b}. For any $U \subset R$ we say that a directed hyperedge \emph{$K^i$ lies in $\Delin(U)$} if $i \in U$ and $K \bs U \neq \varnothing$, i.e.\ if $K^i \in \Delout(R \bs U)$.
Two subsets $U$ and $W$ of $R$ are called \emph{crossing} if
all four sets
$U\setminus W$, $W\setminus U$, $U\cap W$, and $R\setminus (U\cup W)$ are non-empty. A set-function $p:2^R \to {\mathbb Z}$ is a \emph{crossing supermodular} function if
$$p(U) + p(W) \le p(U\cap W) + p(U\cup W)$$
for all crossing sets $U$ and $W$. A directed hypergraph is said to \emph{cover} $p$ if $|\Delin(U)| \ge p(U)$ for all $U \subset R$. Here is the needed result.
\begin{theorem}[Frank, Kir\'{a}ly \& Kir\'{a}ly \cite{FKK03b}]\label{theorem:aloo}
Given a hypergraph $H=(R,\X)$, and a crossing supermodular function $p$,  the hypergraph has an orientation covering $p$ if and only if for every partition $\pi$ of $R$,\\
\indent
\emph{(a)} $\sum_{K\in\X} \min\{1,\rc^\pi_K\} \ge \sum_{\pi_i\in \pi} p(\pi_i)$, and,
\emph{(b)} $\sum_{K\in \X} \rc^\pi_K \ge \sum_{\pi_i\in \pi} p(R\setminus \pi_i)$.
\end{theorem}
We will show every rational solution $x$ to $\eqref{eq:LP-PU}$ can be fractionally oriented to get a feasible solution for $\eqref{eq:LP-PUDir}$, which will complete the proof of \prettyref{theorem:franky}.
Let $M$ be the smallest integer such that the vector $Mx$ is integral. Let $\mathcal{X}$ be a multi-set of hyperedges which contains $Mx_K$ copies of each $K$. Define the function $p$ by $p(U)=M$ if $r \in U \neq R$, and $p(U)=0$ otherwise; i.e.~$p(U)=M$ iff $R \bs U$ is valid.

\begin{claim}
$H=(R,\X)$ satisfies conditions (a) and (b).
\end{claim}
\begin{proof}
  Note $\sum_{\pi_i\in \pi} p(R\setminus \pi_i) = M(r(\pi)-1)$ since
  all parts of $\pi$ are valid except the part containing the root
  $r$. Thus condition (b), upon scaling by $\frac{1}{M}$, is a
  restatement of constraint \eqref{eq:LP-PU2}, which holds since $x$
  is feasible for \eqref{eq:LP-PU}.

  For this $p$, condition (a) follows from (b) in the following
  sense. Fix a partition $\pi$, and let $\pi_1$ be the part of $\pi$
  containing $r$. If $\pi_1 = R$ then (a) is vacuously true, so assume
  $\pi_1 \neq R$. Let $\sigma$ be the rank-2 partition $\{\pi_1,
  R\setminus \pi_1\}$. Then it is easy to check that
  $\min\{1,\rc^\pi_K\}\ge \rc^\sigma_K$ for all $K$, and consequently
  $\sum_{K\in\X} \min\{1,\rc^\pi_K\}\ge\sum_{K\in\X} \rc^\sigma_K$ and
  $\sum_{\pi_i\in \sigma} p(R\setminus \pi_i) = M = \sum_{\pi_i\in
    \pi} p(\pi_i)$. Thus, (a) for $\pi$ follows from (b) for $\sigma$.
\end{proof}

It is not hard to check that $p$ is crossing supermodular. Now using
\prettyref{theorem:aloo}, take an orientation of $\X$ that covers $p$.

For each $K \in \K$ and each $i \in K$, let $n_{K^i}$ denote the
number of the $Mx_K$ copies of $K$ that are oriented as $K^i$, i.e.\
directed towards $i$. So, $\sum_{i\in K} n_{K^i} = Mx_K$. Let
$x'_{K^i} := \frac{n_{K^i}}{M}$ for all $K^i$. Hence $\sum_i x'_{K^i}
= x_{K}$ and $x'$ has the same objective value as $x$.

To complete the proof, we show $x'$ is feasible for
\eqref{eq:LP-PUDir}. Fix a valid subset $U$ and consider condition
\eqref{eq:LP-PUDir2} for a valid set $U$. Note that $p(R \bs U) =
M$. Therefore, since the orientation covers $p$, we get
$$\sum_{K^i \in \Delout(U)} x'_{K^i} = \frac{1}{M}\sum_{K^i \in \Delout(U)} n_{K^i} = \frac{1}{M}\sum_{K^i \in \Delin(R \bs U)} n_{K^i} \ge \frac{1}{M}p(R \bs U) = \frac{1}{M}M = 1$$
as needed.
\end{proof}

\section{Gainless MSTs and Hypergraphic Relaxations}\label{app:app2}

\setcounter{thm_saved}{\value{theorem}}
\setcounter{theorem}{\value{thm_locopt}}
\begin{theorem}[Implicit in \cite{KPT09}] If the
  MST induced by the terminals is gainless, then
  $\OPT\eqref{eq:LP-PU}$ equals the cost of that MST.
\end{theorem}
\setcounter{theorem}{\value{thm_saved}}

\begin{proof}
  Let $\Pi$ be the set of all partitions of the terminal set. As
  before, we let $r(\pi)$ be the rank of a partition $\pi \in \Pi$,
  and we use $E_{\pi}$ for the set of edges in our graph that
  {\em cross} the partition; i.e., $E_{\pi}$ contains all edges
  whose endpoints lie in different parts of $\pi$. Fulkerson's~\cite{Fu71}
  formulation of the spanning tree polyhedron and its dual are
  as follows.

\begin{figure}[h]
  \begin{minipage}{8cm} \begin{align}
      \min \Big\{\sum_{e \in E} c_e x_e: \quad& x \in \R^E_{\ge 0}
      \tag{\ensuremath{\mathcal{M}}}\label{eq:p-sp} \\
      \sum_{e \in E_{\pi}} x_e \geq r(\pi)-1 \quad& \forall \pi
\in \Pi \Big\}
    \end{align} \end{minipage}
  \hfill \vline \hfill
  \begin{minipage}{8cm} \begin{align}
      \max \Big\{\sum_{\pi} (r(\pi)-1)\cdot y_\pi: \quad& y \in
\R^{\Pi}_{\ge 0} \tag{\ensuremath{\mathcal{M}_D}}\label{eq:d-sp} \\
      \sum_{\pi: e \in E_{\pi}} y_{\pi} \leq c_e,\label{eq:d-sp:e}
\quad&\forall e \in E \Big\}
    \end{align}\end{minipage}
%  \caption{\small The unbounded partition relaxation and its dual.}
\end{figure}

  The high-level overview of the proof is as follows. We first give a
  brief sketch of a folklore primal-dual interpretation of Kruskal's
  minimum-spanning tree algorithm with respect to Fulkerson's LP
  (for more information see, e.g., \cite{KPT09}). Running Kruskal's
  algorithm on the terminal set then returns a minimum spanning tree
  $T$ and a feasible dual $y$ to \prettyref{eq:d-sp} such that
  $$ c(T) = \sum_{\pi} (r(\pi)-1) y_{\pi}. $$
  The final step will be to show that, if the returned MST is
  gainless, then the spanning tree dual $y$ is feasible for
  \eqref{eq:LP-PUD}, and its value is $c(T)$ as well.
  Weak duality and the fact that
  the optimal value of \eqref{eq:LP-PU} is at most $c(T)$ imply
  the theorem.

  Kruskal's algorithm can be viewed as a process over time. For each
  time $\tau\geq 0$, the algorithm keeps a forest $T^{\tau}$, and a
  feasible dual solution $y^{\tau}$; initially $T^0=(V, \varnothing)$ and
  $y^0=0$. Let $\pi^{\tau}$ be the partition induced by the connected
  components of $T^{\tau}$. If $T^{\tau}$ is not a spanning tree,
  Kruskal's algorithm grows the dual variable $y_{\pi^{\tau}}$
  corresponding to the current partition until constraint
  \prettyref{eq:d-sp:e} for some edge $e$ prevents any further increase.
  The algorithm then adds $e$ to the partial tree
  and continues. The algorithm stops at the first time
  $\tau^*$ where $T^{\tau^*}$ is a spanning tree.

  Let $T$ be the gainless spanning tree returned by Kruskal, and let
  $y$ be the corresponding dual. We claim that $y$ is feasible for
  \eqref{eq:LP-PUD}. To see this, consider a full component
  $K$. Clearly, the rank contribution $\rc^{\pi^0}_K$ of $K$ to the
  initial partition $\pi^0$ is $|K|-1$; similarly, the final rank
  contribution $\rc^{\pi^{\tau^*}}_K$ is $0$. Every edge that is added
  during the algorithm's run either leaves the rank contribution of
  $K$ unchanged, or it decreases it by $1$. Let $e_1, \ldots,
  e_{|K|-1}$ be the edges of the final tree $T$ whose addition to $T$
  decreases $K$'s rank contribution. Also let
  $$ 0 \leq \tau_1 \leq \tau_2 \leq \ldots \leq \tau_{|K|-1} \leq
  \tau^* $$ be the times where these edges are added. Note that, by
  definition, we must have $c_{e_i}=\tau_i$ for all $i$. We therefore
  have
  \begin{equation}\label{eq:kr-drop}
      \sum_{i=1}^{|K|-1} c_{e_i} = \sum_{i=1}^{|K|-1}\tau_i.
    \end{equation}
  The right-hand side of this equality is easily checked to be
  equal to
  $$ \int_0^{\tau^*} \rc^{\pi^{\tau}}_K d\tau, $$
  which in turn is equal to
  $ \sum_{\pi} \rc^{\pi}_Ky_{\pi}$, by the definition of Kruskal's
  algorithm. It is not hard to see that the left-hand side of
  \eqref{eq:kr-drop} is the drop $\drop_T(K)$ induced by $K$.
  Together with the fact that $T$ is gainless, we obtain
  $$ c_K \geq \drop_T(K) = \sum_{\pi}\rc^{\pi}_Ky_{\pi}. $$
  Now observe that the right-hand side of this equation is the
  left-hand side of \eqref{eq:LP-PUD3}. It follows that $y$ is
  feasible for \eqref{eq:LP-PUD}.
\end{proof}
}{}

\end{document}